%% file: main.tex
\documentclass{ieeeaccess}

\usepackage{amsmath,amssymb,amsfonts}
\usepackage{amsthm,graphicx,mathtools,cancel,float,bm,subcaption}
\usepackage{cite}
\usepackage{bm}
\usepackage{algorithmic}
\usepackage{textcomp}
\usepackage{colortbl}
\usepackage{hyperref}
\usepackage{physics}
\usepackage[shortlabels]{enumitem}
\usepackage{framed}

\def\BibTeX{{\rm B\kern-.05em{\sc i\kern-.025em b}\kern-.08em
    T\kern-.1667em\lower.7ex\hbox{E}\kern-.125emX}}

\DeclarePairedDelimiter\ceil{\lceil}{\rceil}

\newtheorem{theorem}{Theorem}[section]
\newtheorem{definition}{Definition}[section]
\newtheorem{proposition}{Proposition}[section]
\newtheorem{lemma}{Lemma}[section]

\newcounter{mylabelcounter}

\makeatletter
\newcommand{\labelText}[2]{%
#1\refstepcounter{mylabelcounter}%
\immediate\write\@auxout{%
  \string\newlabel{#2}{{1}{\thepage}{{\unexpanded{#1}}}{mylabelcounter.\number\value{mylabelcounter}}{}}%
}%
}
\makeatother

\begin{document}


\title{A Control Architecture for Entanglement Generation Switches in Quantum Networks}
\author{\uppercase{Scarlett Gauthier}\authorrefmark{1,3},
\uppercase{Gayane Vardoyan\authorrefmark{1,3,4}, and Stephanie Wehner}\authorrefmark{1,2,3}}
\address[1]{QuTech, Delft University of Technology}
\address[2]{Kavli Institute of Nanoscience, Delft University of Technology}
\address[3]{Quantum Computer Science, Electrical Engineering, Mathematics and Computer Science, Delft University of Technology}
\address[4]{University of Massachusetts, Amherst}

\markboth
{Gauthier \headeretal: A Control Architecture for Entanglement Generation Switches}
{Gauthier \headeretal: A Control Architecture for Entanglement Generation Switches}

\corresp{Corresponding author: Scarlett Gauthier (email: s.s.gauthier@tudelft.nl).}

\begin{abstract}
       Entanglement between quantum network nodes is often produced using intermediary devices - such as heralding stations - as a resource. When scaling quantum networks to many nodes, requiring a dedicated intermediary device for every pair of nodes introduces high costs. Here, we propose a cost-effective architecture to connect many quantum network nodes via a central quantum network hub called an Entanglement Generation Switch (EGS). The EGS allows multiple quantum nodes to be connected at a fixed resource cost, by sharing the resources needed to make entanglement. We propose an algorithm called the Rate Control Protocol (RCP) which moderates the level of competition for access to the hub's resources between sets of users. We proceed to prove a convergence theorem for rates yielded by the algorithm. To derive the algorithm we work in the framework of Network Utility Maximization (NUM) and make use of the theory of Lagrange multipliers and Lagrangian duality. Our EGS architecture lays the groundwork for developing control architectures compatible with other types of quantum network hubs as well as system models of greater complexity.
\end{abstract}


\titlepgskip=-15pt

\maketitle

\input{LongFormIntroduction}
\input{LongPreliminaries}

\input{LongerAlgorithm}

\input{LongCaseStudy}
\input{LongOpenQuestions}
\input{Proofs}

\bibliographystyle{unsrt}
\bibliography{longRefs}

\EOD

\end{document}

%% file: LongFormIntroduction.tex
\section{Introduction}
\label{section:Intro}

A quantum network enables radically new capabilities that are provably impossible to attain in any classical network \cite{RoadmapQCC}. Examples  include applications such as secure communication \cite{bb84, e91}, secure quantum computing in the cloud \cite{bqc1, bqc2}, and clock synchronization \cite{clockSynchPaper}. Users utilize the end nodes of a network to run applications. The key to unlocking widespread roll-out of these applications is the ability to produce entanglement between these end nodes. 

Prevalent methods for generating entanglement between two quantum nodes that are directly connected by a quantum communication medium (e.g., optical fibers) involve an intermediate device. A prime example is heralded entanglement generation \cite{theoryHeraldCCGZ, theoryHeraldDLCZ} in which the intermediary device is a so-called heralding station. This method of producing entanglement has successfully been demonstrated in many experimental platforms including Color Centers \cite{DoubleClickDiamond, SingleClickDiamond}, Ion Traps \cite{HerEntTrappedIons1, HerEntTrappedIons2}, Atomic Ensembles \cite{HerEntOriginalAE, HerEntSecondAE} and Neutral Atoms \cite{NeutralAtomsHeralded}. As quantum networks continue to scale, it becomes increasingly impractical to maintain direct fiber connections and dedicated heralding stations for every pair of end nodes.

To address this challenge, we propose a scalable quantum network architecture for an Entanglement Generation Switch (EGS), a central hub equipped with a limited number of intermediate devices called resources, a switch, and a processor responsible for managing a scheduling algorithm and sending classical messages to nodes. This central hub enables multiple nodes to share the intermediate devices, significantly reducing the complexity and total resources required for large-scale deployment. While our results apply to an EGS sharing any type of entanglement generation resource, a specific example illustrates how an EGS can operate: Consider quantum network nodes that generate entanglement between them using the so-called single-click bipartite entanglement generation protocol (see e.g \cite{SingleClickDiamond}). In this case the resource(s) to be shared are the heralding station(s). Such stations consist of two input channels connected to a $50/50$ beam splitter, which is then connected by two output channels to a pair of photon detectors that are each connected to a device for processing the measurement outcomes, such as a Field Programmable Gate Array (FPGA). The basic principle of the single-click protocol requires that each network node of the pair locally generates entanglement between a qubit in their local memory and a travelling photon. The photon is sent to a heralding station at which an entanglement swap is attempted on the two photons received; if the entanglement swap is successful, the qubits of the two network nodes will have become entangled. An EGS aims to share one or more heralding stations amongst many connected network nodes. These nodes will still run the single-click protocol, but be limited to using the heralding station needed in the time allocated to them by the EGS. 

A crucial challenge in implementing such an architecture is the efficient allocation of the central hub's resources to different pairs of users in distinct time slots. Similar to classical networking, the allocation process should be driven by user demand for network resources. In the context of quantum networks, this translates to the demand of a user pair $(u_i, \ u_j)$ for entanglement generation at a specific rate or fidelity. Given a set of user demands, the EGS must compose a schedule for the allocation of resources in order to service those demands. In general, the total demand of users may exceed the available resources at the central hub, leading to scheduling and resource allocation challenges.

Here, we introduce the first algorithm for regulating user demand to an EGS, thereby solving this key challenge. Specifically, the algorithm takes as input a vector of rates of entanglement generation demanded by pairs of users and outputs an updated rate vector. The current set of user-originated demands is a measure of competition for EGS resources. We construct the algorithm within the Network Utility Maximization (NUM) framework, wherein the problem of demand regulation is cast as a constrained optimization problem. To solve the problem, we derive the algorithm by using the theory of Lagrange multipliers and Lagrangian duality. These tools, respectively, enable including the constraints together with the objective of the optimization problem and solving for a parameter vector which is the unknown value of the combined problem. Regulating competition for the resources by modifying user demand makes it possible to enforce a notion of fairness in the allocation of resources and maximize resource utilization. Since the algorithm regulates competition by calculating the rates demanded by users, we call it the Rate Control Protocol (RCP).

\subsection{Results Summary}
We make the following contributions:
\begin{itemize}
    \item We characterize (Theorem \ref{thm:CapReg}) the capacity region of the EGS, which is the maximal set of rates at which users can demand entanglement generation such that there exists a scheduling policy under which, on average, the demanded rates do not exceed the delivered rates. The impact of specifying the capacity region is that it delineates which rates can feasibly be serviced by the EGS.  
    \item We prove (Theorem \ref{thm:CapReg}) that under the Maximum Weight Scheduling policy (Definition \ref{def:MaxWeight}) for resource allocation it is possible for the EGS to deliver average rates of entanglement generation that match the requested rates, for any rate vector from within the capacity region. Therefore, an EGS operated with this scheduling policy can achieve throughput optimality as long as the rates demanded by users lie withing the capacity region. To prove the theorem, we use the Lyapunov stability theory of Markov chains. 
    \item We derive the RCP, an algorithm to regulate the rates of bi-partite entanglement generation which pairs of users demand from an EGS. The RCP solves the problem of moderating user competition for EGS resources. The derivation is based on techniques from Network Utility Maximization (NUM) and its quantum network extension (QNUM), where resource allocation in a (quantum) network is modelled as an optimization problem that can be solved using methods from convex optimization theory. 
    
    \item We prove (Theorem \ref{thm:ConvergenceThm}) that the sequence of arrival rate vectors yielded by the RCP converges over time slots to an optimum value, given any feasible rate vector as initial condition. The significance of this result is that if the RCP is used to set the demand rates of entanglement generation over a series of time-slots, the set of demanded rates will approach an optimal value, as long as the initial rate vector supplied to the algorithm is feasible. The proof relies on Lagrange multipliers and Lagrangian duality theory.

     \item Finally, we supply numerical results that support our analysis.
\end{itemize}

\subsection{Related Work}
A quantum network hub that can store locally at least one qubit per linked node and distributes entanglement across these links has been studied \cite{GayaneSwitch2, GuusEDS}. We refer to such a hub as an Entanglement Distribution Switch (EDS). This system differs from our system because the central hub has qubits and/or quantum memories, whereas our system does not. In \cite{GayaneSwitch2} the focus is on assessing the EDS performance in terms of the rate at which it creates $n-$partite entanglements, and in \cite{GuusEDS} the possible rate/fidelity combinations of GHZ states that may be supplied by an EDS \cite{GuusEDS} are studied. 

Maximum Weight scheduling is a type of solution to the problem of resource allocation which is based on assigning resources to sets of users with the largest service backlog. A Maximum Weight scheduling policy was originally presented in \cite{OrignalMaxWeightComms} for resource allocation in classical communication networks and was adapted to the analysis of a single switch for classical networking in \cite{MaxWeightClassicalSwitch}, where it was shown that under this scheduling policy the set of request arrival rates matches the request departure rates (or in other words the policy stabilizes the switch for all feasible arrival rates). In \cite{Vasantam_2022} the capacity region of an EDS, defined as the set of arrival rates of requests for end-to-end multi-partite entanglements that stabilize the switch, is first characterized. Using the Lyapunov stability theory of Markov chains, a Maximum Weight scheduling policy is proposed and shown to stabilize the switch for all arrival rates within the capacity region.
To summarize, in each of the classical network settings and in the EDS setting a Maximum Weight scheduling policy has the merit of achieving a specified performance metric. None of these results are immediately applicable to our system. We demonstrate that such a policy achieves the performance metric of throughput optimality when applied to the EGS by first characterizing the capacity region of the EGS, which has not been done before, and then proving that a Maximum Weight scheduling policy also achieves throughput optimality in our system. 

These results on the analysis of EDS systems constitute the first analytic approaches to resource allocation by a quantum network hub. However, due to the assumption that an EDS locally controls some number of qubits per link, the system has a high technical implementation cost which may not be compatible with near-term quantum networks. Moreover, although these works assume that there is competition between multiple sets of users, the focus is purely on the capacity of the EDS system. Conversely, our analytic contributions apply to EGS quantum network hubs, which have a low technical implementation cost because the hub does not require local control of any qubits or quantum memory. Furthermore, our results extend beyond the analysis of the capacity of the EGS and we propose the RCP as a solution to the problem of moderating competition for the EGS resources. 

In \cite{QuantumCity}, a quantum network topology is studied where user-controlled nodes are connected through a hub known as a Qonnector. The Qonnector provides the necessary hardware for limited end nodes to execute applications in pairs or small groups. A potential configuration of the Qonnector is as an EGS. While \cite{QuantumCity} focuses on assessing the performance of certain applications in this topology, it does not address control policies for the system. In contrast, our work examines control policies for an EGS.

NUM was first introduced in \cite{KellySeminal} and has been widely used to develop and analyze control policies for classical networks \cite{SrikantYing}. It is a powerful framework for designing and analyzing communication protocols in classical networks wherein the problem of allocating resources amongst competing sets of users is cast as a constrained optimization problem. This framework was recently extended to QNUM by \cite{GayaneNUM}. Therein, the authors first develop three performance metrics and use them to catalogue the utility of resource allocation in a quantum network model where each link is associated with a rate and fidelity of entanglement delivery to communicating users. This work does not immediately extend to control policies, as the resource allocations investigated are based on static numerical optimization and need to be recalculated in response to changes in the constraints or sets of users.

In classical networks, probabilistic failures such as loss of a message during transmission or irreconcilable distortion due to transmission over a noisy channel may occur. A serious challenge introduced in the analysis of quantum networks is that in addition to the failure modes of a classical network several new probabilistic failure modes arise that are independent of the state of the network but nevertheless affect its ability to satisfy demands. An example is the probabilistic success in practical realizations of heralded entanglement generation \cite{DoubleClickDiamond, SingleClickDiamond, HerEntTrappedIons1, HerEntTrappedIons2, HerEntOriginalAE, HerEntSecondAE, NeutralAtomsHeralded}. Due to this failure mode, scheduling access to a resource at a certain rate does not guarantee entanglement generation at that rate, thereby complicating the analysis of scheduling.

It is important to distinguish between the concept of rates in classical network control protocols and the notion of rate in the model of a quantum network hub presented here. In classical networks, users transmit \textit{data} at some rate and classical network control protocols, such as the Transmission Control Protocol (TCP), regulate the rate at which users send their data \cite{SrikantYing}. In contrast, in our quantum network hub model, users demand a rate of entanglement generation. However, a significant challenge in developing a control protocol for the EGS is the difference between the rate of attempted entanglement generation and the rate at which entanglement delivery is demanded and delivered to users. Explicitly, in the RCP it is the desired rates of entanglement generation that serve as the controllable parameters moderated by the protocol. 

\begin{figure*}[ht]
    \centering
    \includegraphics[scale=0.43]{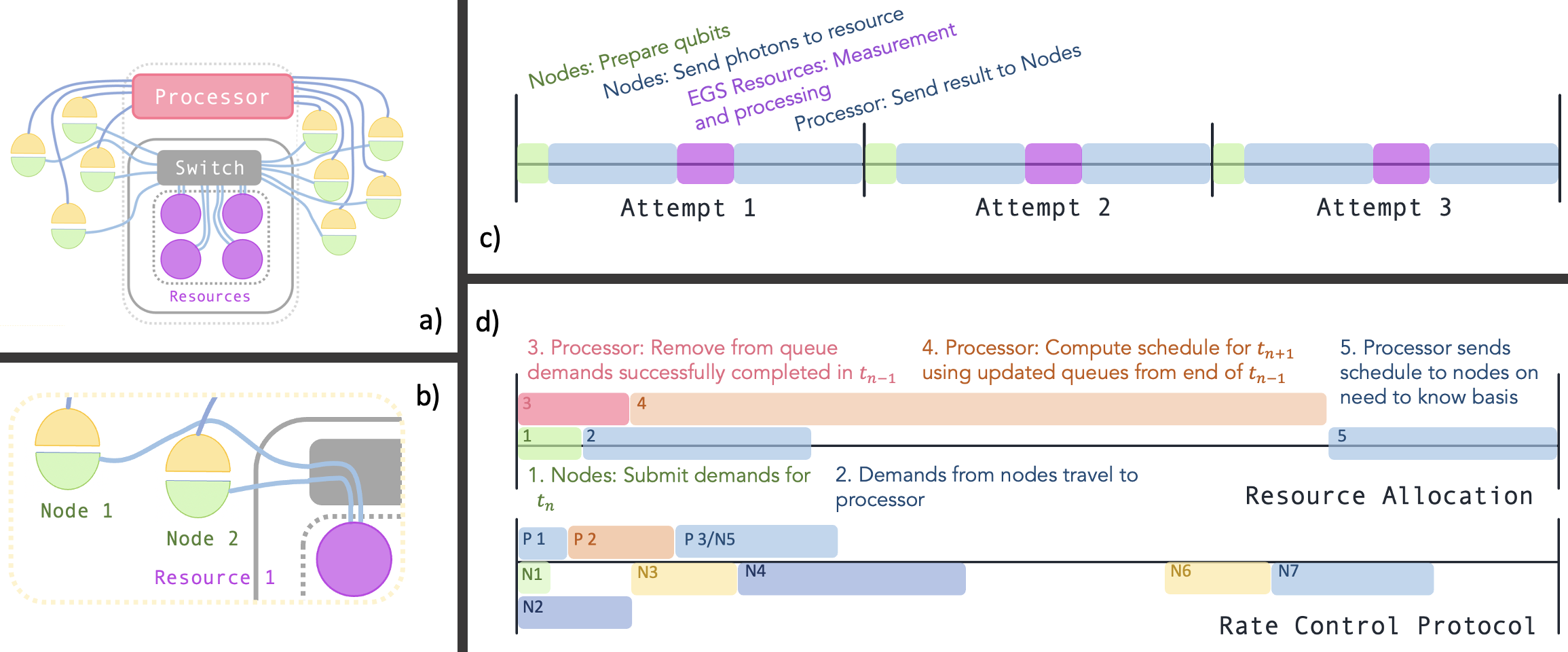}
    \caption{EGS Architecture: a) EGS structure: An EGS with $R=4$ resources connected to $N=9$ nodes. The EGS is controlled by a classical processor and consists of a switch, resources, and physical connections. Nodes have quantum communication channels to the switch and classical communication channels to the processor. b) Resource Allocation: The switch opens connections to link nodes 1, 2 and resource 1. For example, the connections may consist of direct optical fiber paths from the nodes to the switch and from the switch to the resource, via an interface at the switch. This establishes the physical allocation of resource 1 to the communication session of nodes 1, 2 for time slot $t_n$. c) Quantum communication sequence: Node-to-processor communication in time slot $t_n$ with a batch size of three entanglement generation attempts. d) Concurrent classical communication sequences: Nodes and the processor communicate in time slot $t_n$, governing resource allocation and the RCP (see Algorithm \ref{Box:rateControlProt} for RCP details.) }
    \label{fig:BasicSwitchSetup}
\end{figure*}

%% file: LongPreliminaries.tex
\section{Preliminaries}
\label{sec:Preliminaries}
Operation of the EGS requires interactions between the set of quantum network nodes $U$ and the EGS processor with control over $R$ resources. See Fig. \ref{fig:BasicSwitchSetup} a) for an overview of the physical architecture. Below we delineate the process by which pairs of nodes may request $\big{(}$Fig. \ref{fig:BasicSwitchSetup} d)$\big{)}$ and receive $\big{(}$Fig. \ref{fig:BasicSwitchSetup} b) and d)$\big{)}$ resource allocations from the processor. We assume: \begin{itemize}
    \item the EGS operates in a fixed-duration time slotted system where $t_n$ denotes the $n^{th}$ time slot;
    \item timing synchronization between the processor and each node is continuously managed by classical control electronics at the physical layer;
    \item allocation of a single resource to communication session $s$ for one time slot allows for the creation of a maximum of one entangled pair with a success probability of $p_{\text{gen}}$.  A consistent physical model involves a \textit{batched sequence} of attempts, which can be terminated upon the successful creation of an entangled pair or at the end of the time slot. See Fig. \ref{fig:BasicSwitchSetup} c) for an example quantum communication sequence compatible with heralded entanglement generation. 
\end{itemize}
The classical communication sequence repeated in each time slot $t_n$ which governs resource allocation is summarized in Fig. \ref{fig:BasicSwitchSetup} d). In what follows we introduce and explain each step of this communication sequence. The notation introduced throughout this section is summarized in Table \ref{tab:paramInventory}. 

\begin{table*}
    \centering
        \caption{Inventory of notation introduced in Section \ref{sec:Preliminaries}}
    \label{tab:paramInventory}
    \begin{tabular}{|c|c|c|}
         \hline \textbf{Identifier} & \textbf{Description} & \textbf{Domain}\\ \hline
         $U$ & the set of (user operated) quantum network nodes, of cardinality $|U| = N$ & $\mathbb{N}$ \\ \hline
         $R$ & the number of resources controlled by the EGS processor & $\mathbb{N}$ \\ \hline
         $t_n$ & the $n-th$ time-slot of the EGS system  & $\mathbb{N}$ \\ \hline
         $S$ & the set of communication sessions & $\{1, \cdots, N\}^{2}$ \\ \hline
         $\bm{\lambda}(t_n)$ & the vector of target rates of all communication sessions at time $t_n$ & ${\mathbb{R}^{+}}^{|S|}$ \\ \hline
         $a_s(t_n)$ & the number of demands from communication session $s$ in time-slot $t_n$ & $\mathbb{N}$ \\ \hline
         $\bm{q}(t_n)$ & the vector of queues, with components $\big{(}q_s(t_n) \ \forall s\big{)}$ & $\mathbb{N}^{|S|}$ \\ \hline
         $\bm{M}(t_{n})$ & the resource allocation schedule for time-slot $t_n$, & $\mathbb{N}^{|S|}$ \\ 
         & with components $\big{(}M_s(t_n) \ \forall s \big{)}$ & \\ \hline
         $x_s$ & the maximum number of resources that can be allocated to  & $\{1, \cdots, R \}$ \\
         & communication session $s$ in any one time-slot & \\ \hline
         $p_{\text{gen}}$ & the probability a communication session allocated a resource for one time-slot & [0, 1] \\
         & successfully generates entanglement & \\ \hline
         $g_s(t_n)$ & the number of successfully generated entangled pairs by & $\{0, 1, \cdots, M_s(t_n) \}$ \\
         & communication session $s$ in $t_n$ &  \\ \hline
         $\mathcal{C}$ & the set which has the capacity region of the EGS as interior & $\mathbb{R}^{+}$ \\ \hline
         $\lambda_{\text{EGS}}$ & the maximum total rate that can be delivered, on average, by the EGS, $\lambda_{\text{EGS}} = R \cdot p_{\text{gen}}$ & $\mathbb{R}^{+}$ \\ \hline
         $\lambda^{\max}_{\text{gen}, s}$ & the maximum rate the EGS can deliver, on average, to communication & $\mathbb{R}^{+}$ \\ 
         & session $s$, $\lambda^{\max}_{\text{gen}, s} = x_s \cdot p_{\text{gen}}$ & \\ \hline
         $\lambda^{\min}_{s}$ & a minimum acceptable rate of entanglement generation specified by & $\mathbb{R}^{+}$ \\ 
         & communication session $s$ & \\ \hline
         $\lambda_u \phantom{x.}$ & the maximum rate at which each node $u \in U$ can generate & $\mathbb{R}^{+}$ \\
         & and/or make use of entanglement, across all of the sessions that it is involved in & \\ \hline
         \end{tabular}
\end{table*}

\subsection{Demands for resource allocation from nodes to the EGS processor}
\begin{definition}[Target Rate, Communication Session]
\label{def:CommSession}
    Each possible pair of nodes has the potential to require shared bipartite entanglement. To fulfill this need, a node pair $(u_i, u_j)$ requires the processor to allocate a resource. The node pair sets a \textit{target rate} $\lambda_{(i, j)}(t_n)$ once per time slot, which represents the average number of entangled pairs per time slot they aim to generate using one or more EGS resources. A distinct pair of nodes with a non-zero target rate is referred to as a \textit{communication session} and is associated with a unique communication session ID, $s$. The set of communication session IDs, $S$ is defined as follows:
\begin{align}
        \label{eq:set_session_ids}
    S := \big{\{} s &=  (i, j) \ | \ i < j \text{ and } \nonumber \\
    & \lambda_{s}(t_n) > 0,  \forall \ (i, j) \in \{1, \cdots, N \}^{2} \big{\}}
\end{align} 
where $N= |U|$ is the total number of network nodes with connections to the EGS. 
\end{definition}

Henceforth each pair of nodes will be identified by its communication session id $s$. The target rates of all communication sessions in time-slot $t_n$ can be written as a vector $\bm{\lambda}(t_n) \in \mathbb{R}^{|S|}$, the $s^{th}$ component of which is labelled by communication session ID $s$ as $\lambda_s(t_n)$. 

A rate of entanglement generation is the service demanded by each communication session from the EGS. To address the difference between the desired rate and the rate at which a communication session requires resource allocation to achieve that rate, we establish the following model for demand, which is compatible with a discrete time scheduling policy.  

\begin{definition}[Demand]
    \textit{Demands} for resources are requests made by communication session $s$ to obtain a single entangled pair. The number of demands $a_s(t_n)$ submitted by session $s$ at time slot $t_n$ depends on its target rate $\lambda_s(t_n)$. If $\lambda_s(t_n) > 1$, then communication session $s$ first submits $\lfloor \lambda_s(t_n) \rfloor$ demands. For a communication session $s$ with $0 \leq \lambda_s(t_n) \leq 1$, or to account for the remaining part of the rate for any session with $\lambda_s(t_n) > 1$, each communication session randomly generates demands by sampling from a Bernoulli distribution with a mean equal to $\lambda_s(t_n) - \lfloor \lambda_s(t_n) \rfloor$, so that in general the submitted demands satisfy a (shifted) Bernoulli distribution, $a_s(t_n) \sim \text{Bernoulli}\big{(}\lambda_s(t_n) - \lfloor \lambda_s(t_n) \rfloor\big{)} + \lfloor \lambda_{s}(t_n) \rfloor$.  
\end{definition}

\begin{definition}[Designated Communication Node, Secondary Node]
\label{def:designatedSecondary}
    One of the nodes of every communication session is marked as the \textit{designated communication node} for communicating the entanglement requests to the switch. The terms designated communication node and \textit{secondary node} are used to refer to the two nodes of a communication session.
\end{definition}

\subsection{Processing demands for resource allocation}

\begin{definition}[Queue]
\label{def:queue}
    When the processor receives a demand, it is added to one of $|S|$ \textit{queues}, one for each communication session. The set of demands received by the processor by time-slot $t_n$ and not yet satisfied is captured by the queue vector $\mathbf{q}(t_n) \in \mathbb{N}^{|S|} = (q_s(t_n) \ \forall s)$, where the component $q_s(t_n)$ is the queue of communication session $s$ at time $t_n$. Each queue processes demands in first in first out order. As all demands are identical, we interchangeably use $q_s(t_n)$ to refer to both the queue length of communication session $s$ in time slot $t_n$ and the queue itself.
\end{definition}

\begin{definition}[(Demand-Based) Schedule]
A resource allocation \textit{schedule} is a vector $\bm{M}(t_{n+1}) \in \mathbb{N}^{|S|}$ calculated by the EGS processor in time slot $t_{n}$ determining the assignment of the resources for time slot $t_{n+1}$. A single session $s$ may be allocated the use of multiple resources, up to a maximum number $x_{s}$ set by the EGS which does not exceed $R$, the total number of resources controlled by the EGS. For every session $s \in S$ the entry \begin{equation}\label{eq:schedEntries} M_s(t_{n+1}) \in \{0, \ 1, \ \cdots, \ x_{s} \}\end{equation} corresponds to the number of resources assigned to $s$ for the entire duration of time slot $t_{n+1}$. A \textit{demand based} schedule is based on the vector of all queues, $\bm{q}(t_{n})$, as it stands before new demands are registered in $t_{n}$, and satisfies
\begin{align}
    \label{eq:schedFullRestrict}
    \sum_s M_s(t_{n+1}) &\leq \min \big{(}\sum_s q_s(t_{n}), \ R \big{)}, \\
    \label{eq:schedSingleSessionRestrict}
    0 \leq M_s(t_{n+1}) & \leq \min\big{(}q_{s}(t_{n}), \ x_{s}\big{)} \leq R,  \ \forall \ s.
\end{align}
\end{definition}

Each node of a communication session $s$ requires a physical connection to the EGS switch. A single physical connection, such as an optical fiber, can be used for this purpose. To enable multiple connections between a node and the switch, options include the use of optical multiplexers over a single fiber or utilizing multiple fibers within a fiber bundle. The parameters $\big{(}x_s \ \forall s \big{)}$ are motivated by situations where the number of physical connections that can be dedicated to service communication session $s$ are limited. 

\begin{definition}[Maximum Weight Scheduling]
\label{def:MaxWeight}
The set $\mathcal{M}$ of feasible demand based schedules at time slot $t_n$ contains all vectors $\mathbf{M}'(t_{n+1}) \in \mathbb{N}^{|S|}$ satisfying (\ref{eq:schedEntries}), (\ref{eq:schedFullRestrict}), and (\ref{eq:schedSingleSessionRestrict}). The EGS processor selects a \textit{maximum weight schedule} $\mathbf{M}(t_{n+1}) \in$ $\mathcal{M}$ from the feasible schedules for the following time slot by solving for 
\begin{equation}
\label{eq:MaxWeightSched}
    \mathbf{M}(t_{n+1}) \in \text{arg}  \ \underset{\mathbf{M}'}{\text{max }} \sum_s q_s(t_{n}) M'_s(t_{n+1}). 
\end{equation}
In words, the schedule is selected from the set of feasible schedules by first solving for the subset of schedules that allocate resources to the sessions with the largest number of queued demands. If that subset contains more than one schedule, a schedule is randomly selected from the subset. 
\end{definition}

By the end of $t_n$, the schedule for $t_{n+1}$ has been computed by the processor and broadcast to the nodes. If the schedule allocates use of a resource to communication session $s$ for $t_{n+1}$, the users of $s$ utilize the allocated resource to make a batch of entanglement generation attempts over the duration of $t_{n+1}$. The demand at the front of queue $s$ is only marked as served once both a resource has been allocated and the users of $s$ have successfully generated entanglement. Hence the dynamics of each queue are given by,
    \begin{equation}
    \label{eq:queuesUpdate}
    q_s(t_{n+1}) = [q_s(t_n) + a_s(t_n) - g_s(t_n)]^{+} \ \forall \ s, 
    \end{equation}
    
    \noindent where $[z]^{+} = \max(z, 0)$, and $g_{s}(t_n)$ is the number of successfully generated entangled pairs by $s$ during $t_n$. In words, every subsequent time slot the demands that arrived in the previous time slot are added to the queue and those that were scheduled and successfully resulted in the generation of an entangled pair are removed from the queue. The updated queue is always of non-negative length since the number of successfully generated entangled pairs is a sample of a binomial random variable where the number of trials is the number of resources allocated to $s$, $M_{s}(t_n) \big{(}\leq q_s(t_n)\big{)}$, and the trial success probability is $p_{\text{gen}}$, 
\begin{equation*}
    g_{s}(t_n) \sim \text{Bin} \big{(}M_{s}(t_n), \ p_{\text{gen}}\big{)}.
\end{equation*}

\begin{definition}[Supportable rate]
\label{def:supportable}
    The arrival rate vector $\bm{\lambda}(t_n) \in \mathbb{{R}^{+}}^{|S|} = \big{(}\lambda_s(t_n) \ \forall \ s\big{)}^{\text{T}}$ is \textit{supportable} if there exists a schedule under which,
    \begin{equation}
    \label{eq:supportableArrivals}
    \underset{Q \rightarrow \infty}{\text{lim}} \ \underset{n \rightarrow \infty}{\text{lim}} \text{P} \big{(} |\bm{q}(t_n)| \geq Q \big{)} = 0,
    \end{equation}
    where $|\bm{q}(t_n)| := \underset{s}{\Sigma} |q_s(t_n)|$ is the sum of the number of demands in the queue of each session in time slot $t_n$. That is, $\bm{\lambda}(t_n)$ is supportable if the probability that the total queue length becomes infinite is zero.
\end{definition}

\begin{definition}[Capacity Region]
    The \textit{capacity region} of an EGS is the set of arrival rate vectors that are supportable by the EGS. For each rate vector $\bm{\lambda}$ in the capacity region, there exists some scheduling routine such that an EGS operating under that scheduling algorithm can support the rate vector $\bm{\lambda}$.  
\end{definition}

If the rate vector $\bm{\lambda}$ falls outside the capacity region, the EGS cannot support it under any scheduling algorithm, leading to unpredictable performance. The goal of moderating the rate vector through the Rate Control Protocol (RCP) is twofold: first, to keep it within the capacity region, and second, to maximize resource utilization by saturating the capacity region,  thus fully leveraging the potential of the EGS to facilitate entanglement generation.

\begin{theorem}[Capacity Region]
\label{thm:CapReg}
Let $x_{s}$ be the maximum number of resources that can be allocated to a session $s$ per time slot. For each resource, $p_{\emph{\text{gen}}}$ is the probability that a communication session allocated the resource for one time slot will successfully create an entangled pair. The capacity region of an EGS with R resources is the set of rate vectors $\bm{\lambda} \in \text{\emph{Int}} \mathcal{C}$, where $\mathcal{C}$ is defined as:
\begin{equation}
    \label{eq:CapReg}
    \mathcal{C} = \big{\{} \bm{\lambda}: \bm{\lambda}\geq \bm{0}, \ \sum_s \lambda_s \leq \lambda_{\text{\emph{EGS}}}, \text{ and } \lambda_{s} \leq \lambda^{\max}_{\text{\emph{gen}}, s} \ \forall \ s \in  S \},
\end{equation}
$\lambda_{\text{\emph{EGS}}} = R \cdot p_{\text{\emph{gen}}}$ and $\lambda^{\max}_{\text{\emph{gen}}, s} = x_{s} \cdot p_{\text{\emph{gen}}}$. Moreover, maximum weight scheduling (Definition \ref{def:MaxWeight}) is throughput optimal and supports any rate vector $\bm{\lambda} \in \emph{\text{Int}} \mathcal{C}$. For proof, see Section \ref{subsec:proofCapReg}. 
\end{theorem}

The first requirement of $\mathcal{C}$ states that all request rate vectors must be positive, meaning every component of the rate vector must be positive or zero $(\bm{\lambda} \geq \bm{0} \Leftrightarrow \lambda_s \geq 0 \ \forall s \in S$). The second requirement enforces that the total rate of entanglement requested from the EGS, $\underset{s}{\sum} \lambda_s$, cannot exceed the total average service rate of the EGS, $R \cdot p_{\text{gen}}$. The final requirement states that the request rate $\lambda_s$ of any communication session $s$ must not exceed the maximum average service rate that can be allocated to the communication session, $x_{s} \cdot p_{\text{gen}}$.

\subsection{Constraints}
We assume that there are two types of constraints on the sequence of target rates set by a session. The first is a minimum rate of entanglement generation $\lambda^{\min}_s$; below this rate, session $s$ cannot obtain sufficient entangled pairs within a short enough period of time in order to enable its target application. The second constraint $\lambda_{u} \ \forall u \in U$ is an upper limit on the rate at which each node $u$ can generate and/or make use of entanglement across all of the sessions that it is involved in. This parameter can capture a range of technical limitations of the quantum nodes, including a limited rate of entanglement generation or a limited speed of writing generated entanglement to memory, hence temporarily decreasing the availability of the node for engaging in further entanglement generation immediately following the successful production of a pair.

%% file: LongerAlgorithm.tex
\section{RCP Algorithm}
\label{sec:Algorithm}

An algorithm moderating competition for EGS resources enables the possibility of introducing a notion of fairness in how resources are allocated amongst competing communication sessions and ensuring that the resources are fully utilized. We consider a situation where the rate vector produced by any such algorithm is constrained by the maximum service rate of the switch, as described by the capacity region $\mathcal{C}$, as well as the node or user level constraints described by $\lambda_{u} \ \forall u$ and $\lambda^{\min}_{s} \ \forall s$. In the framework of NUM, we pose an optimization problem where each communication session $s$ is associated with a utility function $f_s(\lambda_s(t_n)): \mathbb{R} \mapsto \mathbb{R}$, which encodes the benefit $s$ derives from the rate vector $\bm{\lambda(t_n)}$. We apply the theory of Lagrange multipliers and Lagrangian duality (see \cite{Bertsekas} for detailed coverage) to formulate and analyze the optimization problem. We then derive the RCP (Algorithm \ref{Box:rateControlProt}) as the solution to this problem.

The primal problem is to maximize the aggregate utility or the total benefit that users derive from the EGS by maximizing the sum of the utility functions, including the constraints by the use of Lagrange multipliers. The dual problem is to determine an optimal vector of Lagrange multipliers. In the case where there is no duality gap \cite{Bertsekas}, a solution to the dual problem is equivalent to a solution of the primal problem. The vector of Lagrange multipliers $\bm{p}(t_{n+1}) = \big{(} p_c(t_{n}),  p_u(t_{n}) \ \forall u \big{)} \in {\mathbb{R}^{+}}^{(1 + N)}$,  with components for the processor and each node, is denoted as the price vector in our algorithm and serves as a measure of the competition for resources amongst the communication sessions. Define $S(u) : = \{ s :  u \in s\} \subseteq S$ to be the subset of communication sessions in which node $u$ participates. In each communication session one node is designated to communicate demand to the switch and the other node is secondary (see Definition \ref{def:designatedSecondary}). Note that $u \in s \Leftrightarrow s \in S(u)$. The feasible rate region of the communication session $s$ is,
\begin{equation}
\label{eq:SessionRateRegion}
    \Lambda_s : = \{\lambda_s :  \lambda^{\min}_{s} \leq \ \lambda_s \leq \lambda^{\max}_{\text{gen}, s}\}  \ \forall \ s,
\end{equation}
and the feasible region for a rate vector $\bm{\lambda}$ is,
\begin{equation}
\label{eq:rateRegion}
    \Lambda = \underset{s}{\bigcup} \Lambda_s.
\end{equation}
 We make the following two assumptions on the utility function $f_s$ of each communication session $s$:
\begin{itemize}
    \item [\textbf{A1:}] On the interval $\Lambda_s = [\lambda^{\min}_{s},  \lambda^{\max}_{\text{gen}, s}]$ the utility functions $f_s$ are increasing, strictly concave, and twice continuously differentiable; 
    \item [\textbf{A2:}] The curvatures of all $f_s$ are bounded away from zero on $\Lambda_s$. For some constant $\alpha_s >0$, 
    \begin{align*}
         -f^{''}_s(\lambda_s) &\geq \frac{1}{\alpha_s}  > 0 \ \forall \ \lambda_s \in \Lambda_s.
    \end{align*}
\end{itemize}

To ensure feasibility and satisfy the Slater constraint qualification \cite{Bertsekas}, in addition to assumptions A1 and A2 it is necessary that the rate vector with components equal to the minimal rates of each communication session is an interior point of the constraint set,
\begin{align}
    \label{eq:Slater1}
    \sum_{s} \lambda^{\min}_{s} &< \lambda_{\text{EGS}}; \\
    \underset{s \in S(u)}{\sum} \lambda^{\min}_{s} &< \lambda_{u} \ \forall \ u. \label{eq:Slater2}
\end{align}

\begin{framed}
\labelText{{ Algorithm 1: Rate Control Protocol (RCP)}}{} \\
\label{Box:rateControlProt}
\rule{0.98\textwidth}{1pt}\\
\vspace{-10pt}
\phantom{x}
\\
\textit{Processor's Algorithm}: At times $t_n = 1, \ 2, \cdots, $ the processor:
\begin{enumerate}
    \item receives rates $\lambda_{s}(t_n)$ from all communication sessions $s \in S$;
    \item computes a new central price, 
    \begin{align}
        \label{eq:CentralPriceUpdate}
    p_{c}(t_{n + 1}) = \big{[} & \frac{1}{\lambda_{\text{EGS}}}\sum_{s} q_s(t_n) \nonumber \\
    &+ \theta_c \big{(}\sum_{s}\lambda_s(t_n) - \lambda_{\text{EGS}} \big{)}\big{]}^{+},
    \end{align}
    where $\theta_c$ is a constant step-size for the central price;
    \item broadcasts the new central price $p_c(t_{n+1})$ to all communication sessions $s \in S$. 
\end{enumerate}

\textit{Network Node u's Algorithm}: At times $t_n = 1, \ 2, \cdots, $ network node $u$:
\begin{enumerate}
    \item marks the subset of communication sessions $\text{COMM}(u) \subseteq S(u)$ involving node $u$ for which it is the designated communication node;
    \item receives from every secondary node $u'$ the price $p_{u'}(t_n)$ for each communication session $s = (u, u') \in \text{COMM}(u);$
    \item computes a new node price, 
    \begin{align}
    \label{eq:userPriceUpdates} p_u(t_{n+1}) = \big{[} & \frac{1}{\lambda_{u}} \underset{s \in S(u)}{\sum}q_s(t_n) \nonumber \\
    &+ \theta_u \big{(} \underset{s \in S(u)}{\sum} \lambda_{s}(t_n) - \lambda_{u}\big{)}\big{]}^{+}, \ \forall \ u,
    \end{align}
    where $\theta_u \ \forall u$ is a constant step-size for each node, which may be fixed or differ from node to node;
    \item communicates the new price $p_u(t_{n+1})$ to the communication node from every communication session $s \in S(u) \setminus \text{COMM}(u)$ in which $u$ is a secondary node;
    \item receives from the switch the central price $p_c(t_{n+1})$;
    \item computes the new rate for every communication session $s \in \text{COMM}(u)$,\small
\begin{equation}\label{eq:formalAlgOptimalRates}
    \lambda_{s}(t_{n+1}) = \Bigg{[} \Big{(} \frac{\text{d} f_{s}}{\text{d}  \lambda_s}\Big{)}^{-1} \big{(}\bm{p}(t_{n+1})\big{)} \Bigg{]}^{\lambda^{\max}_{\text{gen}, s}}_{\lambda^{\min}_{s}},
\end{equation} \normalsize
where $[z]^{M}_{m} = \max \big{(}\min(z, \ M) , m \big{)}$ and $\bm{p}(t_i) = \big{(}p_c(t_i), \ p_u(t_i) \ \forall \ u\big{)}$ is the vector of prices pertaining to time slot $t_i$;
    \item communicates the new rate $\lambda_{s}(t_{n+1})$ to the EGS processor, for every communication session $s \in \text{COMM}(u)$. 
\end{enumerate}

\end{framed}

\subsection{Derivation}
\label{subsec:AlgDerivation}
Formally, the RCP yields rate vectors which solve the
\textit{Primal Problem:}
\begin{equation}\label{eq:session_utility}
    \underset{\bm{\lambda}\in \Lambda }{\max} \ F(\bm{\lambda}) : =  \sum_s f_{s}(\lambda_s), \\
\end{equation}
subject to,
\begin{align}
    \label{eq:SwitchConstraint}
    \sum_s \lambda_s &\leq \lambda_{\text{EGS}} \\
    \label{eq:UserConstraint}
    \sum_{s \in S(u)} \lambda_s &\leq \lambda_u \ \forall  u.
\end{align}

The Lagrangian, which includes the constraints (\ref{eq:SwitchConstraint}), (\ref{eq:UserConstraint}) with a vector of Lagrange multipliers $\bm{p} = (p_c, p_{u} \ \forall u) \geq \bm{0}$ together with the objective function (\ref{eq:session_utility}), is given by
\begin{align}
   \label{eq:FullLagrangian}
   L(\bm{\lambda}, \bm{p}) =  & \sum_{s} f_s(\lambda_s) - p_c \bigg{(} \sum_{s} \lambda_s - \lambda_{\text{EGS}} \bigg{)} \nonumber \\
   & - \underset{u}{\sum}p_u \bigg{(}\underset{s \in S(u)}{\sum} \lambda_s - \lambda_u\bigg{)}.
\end{align}

We identify that the problem is separable in the communication sessions, $S$, and re-write the Lagrangian in separable form,
\begin{align}
\label{eq:SepLagrangian}
    L(\bm{\lambda}, \bm{p}) = \sum_s l_s(\lambda_s) + p_c \lambda_{\text{EGS}} + \sum_{u} p_u \lambda_{u},
\end{align}
where $l_s(\lambda_s)$ is defined as \begin{equation*}
    l_s(\lambda_s) := f_s(\lambda_s) - \lambda_s p_c -  \lambda_s \underset{u \in s}{\sum} p_u,
\end{equation*} and we make use of the equivalence \begin{align*}
  \sum_{u} p_u  \sum_{s \in S(u)} \lambda_s = \sum_{s} \lambda_s \sum_{u \in s} p_u.  
\end{align*}

A rate vector $\bm{\lambda}^{*}$ is a local maximum of (\ref{eq:session_utility}) if it satisfies the optimality condition \cite{Bertsekas}, 
\begin{equation}
\label{eq:optCond}
    \grad_{\lambda_s} F(\bm{\lambda}^{*})^{T} (\bm{\lambda} - \bm{\lambda}^{*}) \leq 0 \ \forall \bm{\lambda} \in \ \Lambda.
\end{equation} If moreover $F(\bm{\lambda})$ is concave over $\Lambda$, then (\ref{eq:optCond}) is also sufficient for $\bm{\lambda}^{*}$ to maximize $F(\bm{\lambda})$ over $\Lambda$ \cite{Bertsekas} (it is also a global maximum).

To obtain a $\bm{\lambda}^{*}$ satisfying both the optimality condition (\ref{eq:optCond}) and the constraints (\ref{eq:SwitchConstraint}), (\ref{eq:UserConstraint}) we set the gradient with respect to rate of each communication session of the Lagrangian to zero, 
\begin{equation*}
    \grad_{\lambda_s} L = \sum_s \frac{\text{d} l_s (\lambda_s)}{\text{d} \lambda_s} = 0.
\end{equation*} The maximization in the primal problem (\ref{eq:session_utility}) is constrained to the feasible rate region defined by (\ref{eq:SessionRateRegion}), (\ref{eq:rateRegion}). To restrict solutions to the problem domain, any $\tilde{\bm{\lambda}}^{*} \cancel{\in} \Lambda$ is projected component-wise so that $\tilde{\lambda}^{*}_{s} \mapsto \lambda^{*}_{S} \in \Lambda_s \ \forall s$. With the assumptions in (\ref{eq:Slater1}), (\ref{eq:Slater2}) there exists at least one set of Lagrange multipliers \cite{Bertsekas}. In terms of a given vector of Lagrange multipliers $\bm{p}$, an optimal rate vector $\bm{\lambda}^{*}$ satisfies,
\begin{equation}\label{eq:formalOptimalRates}
    \lambda^{*}_{s} = \Bigg{[} \bigg{(} \frac{\text{d} f_{s}}{\text{d}  \lambda_s} \bigg{)}^{-1} (\textbf{p}) \Bigg{]}^{\lambda^{\max}_{\text{gen}, s}}_{\lambda^{\min}_{s}} \ \forall \ s,
\end{equation}
where $[z]^{M}_{m} = \max \bigg{(}\min\big{(}z, \ M\big{)} , m \bigg{)}$. To obtain a $ \bm{\lambda}^{*}$, it remains to obtain a vector of Lagrange multipliers. 

An optimal vector $\bm{p}^{*}$ of Lagrange multipliers is a solution to the \textit{Dual Problem:} 
\newline \indent Select $\bm{p} = (p_c, \ p_u \ \forall \ u)$ so as to achieve,
\begin{equation}
\label{eq:dualProblem}
    \underset{\bm{p} \geq \textbf{0}}{\text{inf }} D(\textbf{p}),
\end{equation}
\noindent
where the dual objective function $D(\textbf{p})$ is defined as,
\begin{equation}
\label{eq:dualObjective}
    D(\textbf{p}) = \underset{\bm{\lambda}\in \Lambda}{\text{sup}} \ L(\bm{\lambda}, \bm{p}). 
\end{equation}

With assumptions A1, A2 and (\ref{eq:Slater1}, \ref{eq:Slater2}), the problem satisfies the Slater constraint qualification and has no \textit{duality gap} \cite{Bertsekas}, meaning a solution to the dual problem is also a solution to the primal problem. Define $\bm{\lambda}^{*}$ to be a rate vector that maximizes $L(\bm{\lambda}, \ \bm{p})$. A vector of Lagrange multipliers $\bm{p}^{*}$ is an optimal solution to the dual problem if it satisfies the optimality condition, 
\begin{equation}
    \label{eq:dualOptimalityCond}
    \grad_{p} D(\bm{p}^{*})^{T} (\bm{p} - \bm{p^{*}}) \geq 0 \ \forall \bm{p} \geq 0.  
\end{equation}

Gradient projection is a type of algorithm where in order to solve an optimization problem such as the dual problem, (\ref{eq:dualProblem}), with respect to a vector $\bm{p}$, one starts by selecting some initial vector $\bm{p}(0)$ and iteratively adjusting $\bm{p}(t_n) \mapsto \bm{p}(t_{n+1}) $ by making steps in the opposite direction of the gradient of the objective function. We introduce a vector of step-sizes $\bm{\theta} = (\theta_c, \theta_u \ \forall u) \in \mathbb{R}^{1+N}$.  The components of $\grad_{p} D (\bm{p} )$ are,

\begin{align}
    \label{eq:dualGradC}
    \frac{\partial D(\bm{p})}{\partial p_c} &= - \ \bigg{(} \sum_{s} \lambda^{*}_s - \lambda_{\text{EGS}} \bigg{)}; \\
    & \nonumber \\
    \label{eq:dualGradU}
    \frac{\partial D(\bm{p})}{\partial p_u} &= - \ \bigg{(} \underset{s \in S(u)}{\sum} \lambda^{*}_{s} - \lambda_u \bigg{)} \ \forall \ u. 
\end{align}

An implementation of the gradient projection algorithm is to iteratively adjusting the Lagrange multipliers according to,

\begin{align}
    \label{eq:outlineCentralPrice}p_c(t_{n+1}) &= \big{[} p_{c}(t_n) +  \theta_{c} \big{(}\sum_{s} \lambda^{*}_{s}(t_n) - \lambda_{\text{EGS}}\big{)}\big{]}^{+}; \\
    & \nonumber \\
    \label{eq:outlineUserPricesSession}
    p_{u}(t_{n+1}) &= \big{[}p_{u}(t_n)  + \theta_u \big{(}\underset{s \in S(u)}{\sum} \lambda^{*}_{s}(t_n) - \lambda_u\big{)}\big{]}^{+}, \ \forall \ u,
\end{align}
\noindent
where $\lambda^{*}_{s}(t_n) = \lambda^{*}_{s}\big{(}\bm{p}(t_n)\big{)}$ is given by inputting the vector of Lagrange multipliers in (\ref{eq:formalOptimalRates}). An implementation of the algorithm necessitates identifying parameters in the system that correspond to the components of the vector of Lagrange multipliers. We note that the centralized price $p_c(t_n)$ and the user prices $p_u(t_n) \ \forall \ u$, have respectively, the same dynamics as the total queue lengths and the sum total of the session queue lengths in which user $u$ participates (\ref{eq:queuesUpdate}). Therefore, we make the following identifications,
\begin{align*}
    p_c(t_n) & \leftrightarrow \frac{1}{\lambda_{\text{EGS}}} \sum_s q_s(t_n); \\
    p_u(t_n) & \leftrightarrow \frac{1}{\lambda_u} \sum_{s \in S(u)} q_s(t_n) \ \forall u. 
\end{align*}
Note that these identifications are not unique, since the only strict criteria on the identification is that the queue dynamics generated by (\ref{eq:queuesUpdate}) match the dynamics of (\ref{eq:outlineCentralPrice}) and (\ref{eq:outlineUserPricesSession}), whereas the scaling is arbitrary. For more information on the interpretation of Lagrange multipliers as prices in communication networks, see \cite{KellySeminal, SrikantYing}. 
\subsection{Convergence}
The RCP is a gradient projection algorithm with constant step-sizes from the vector $\bm{\theta} \in \mathbb{R}^{1 + N} = (\theta_c , \theta_u \ \forall u)$. Establishing that the algorithm converges is crucial to ensure that it yields solutions that effectively address the problem it is designed to solve. To establish convergence, we follow a similar treatment as in \cite{FlowControlI}. 

\begin{theorem}[RCP Convergence]
\label{thm:ConvergenceThm}
Suppose assumptions \href{asAl}{A1} and \href{asA2}{A2} and the constraints (\ref{eq:Slater1}, \ref{eq:Slater2}) are satisfied and each of the the step-sizes $\theta_r \in \{ \theta_c, \theta_u \ \forall u \}$ satisfies $\theta_{r} \in (0,   2/\overline{\alpha} |S|)$, where $\overline{\alpha} = \underset{s \in S}{\max} \  \alpha_s$ with $\alpha_s$ the curvature bound of assumption \href{asA2}{A2}, and $|S|$ is the number of communication sessions. Then, starting from any initial rate $\bm{\lambda}(0) \in \Lambda$ and price $\bm{p}(0) \geq \bm{0}$ vectors, every accumulation point $\big{(}\hat{\bm{\lambda}},  \hat{\bm{p}}\big{)}$ of the sequence over time slots $\{\big{(}\bm{\lambda}(t_n) ,  \bm{p}(t_n)\big{)} \}_{t_n}$ generated by the RCP is primal-dual optimal. For proof, see Section \ref{subsec:ConvergenceProof}.
\end{theorem}

%% file: LongCaseStudy.tex
\section{Case Study}
\label{sec:CaseStudy}

\begin{figure*}[ht]
    \centering
    \includegraphics[scale=0.36]{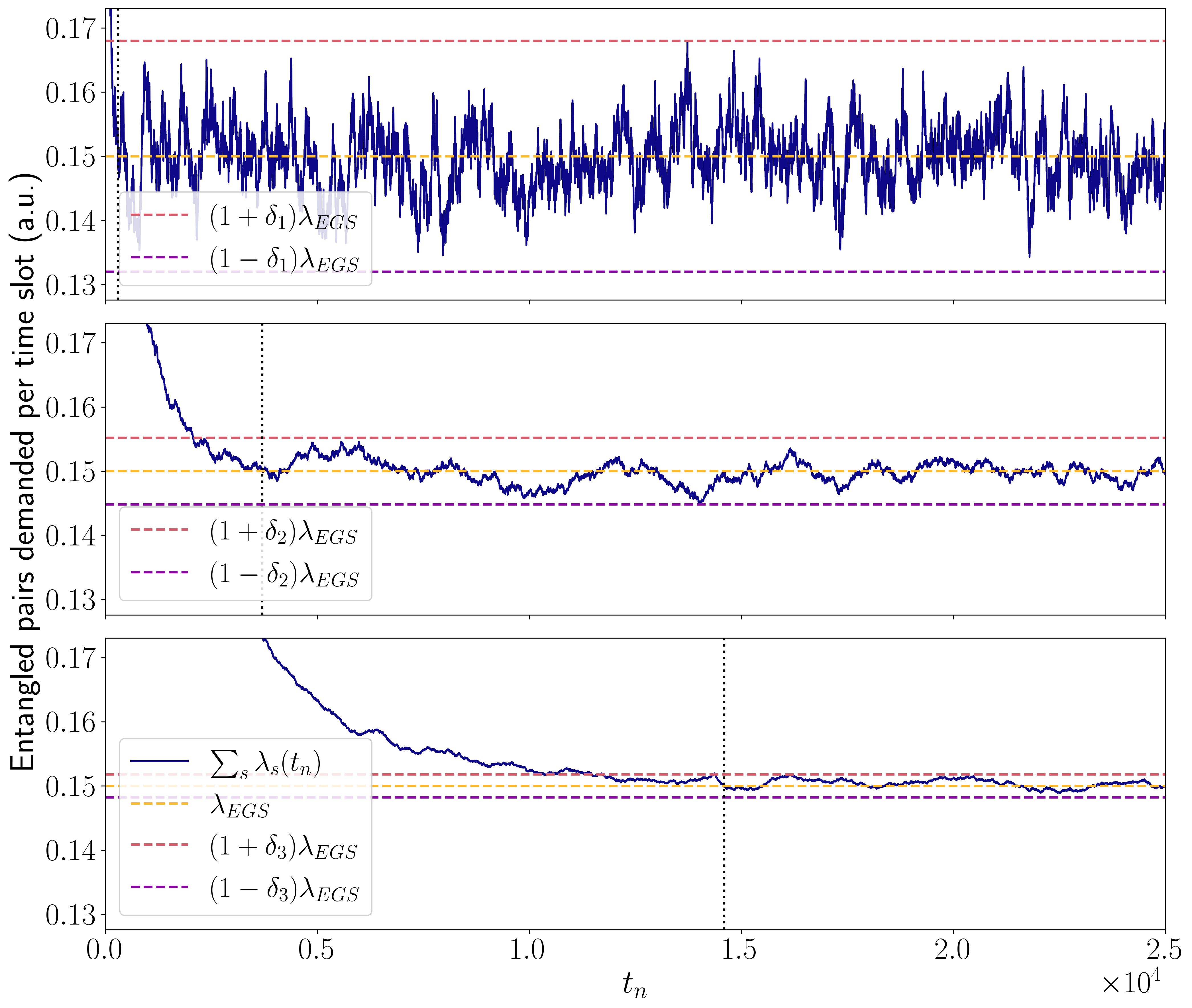}
    \caption{The RCP drives the sum of the demanded rates of entanglement generation across all communication sessions, $\underset{s}{\Sigma} \lambda_s(t_n)$, to converge with respect to the sequence of time slots to the maximum average entanglement generation rate of the EGS, $\lambda_{\text{EGS}}$. The EGS has $R=3$ resources, the probability of entanglement generation is $p_{\text{gen}} = 0.05$, and the EGS is connected to $N=20$ (top), $N=50$ (middle) and $N=100$ (bottom) nodes. The total number of communication sessions served are $|S| = 19, \ 123, \ \text{ and } 495$ in the top, middle, and bottom plots, respectively. Black dotted lines indicate the convergence times, $\Delta \tau$. The observed values for the tightness of convergence, $\delta$, are $ \delta_1 = 0.12$, $ \delta_2= 0.035$ and $ \delta_3 = 0.012$. Step-sizes $(\theta_c, \theta_u \ \forall u)$ were all $1/(40 \cdot \lambda_{\text{EGS}})$.}
    \label{fig:convStudy}
\end{figure*}

To illustrate use of the RCP we associate a log utility function with each session, 
\begin{equation}
    f_s(\lambda_s) = \log(\lambda_s) \ \forall \ s \in S.
\end{equation}
Log utility functions are suitable when throughput is the target performance metric, and a set of sessions all employing log utility functions will have the property of proportional fairness. In such a system, if the proportion by which one session rate changes is positive, there is at least one other session for which the proportional change is negative \cite{SrikantYing}. For compatibility with Theorem \ref{thm:ConvergenceThm} note that log utility functions satisfy A1, and A2 is satisfied with $\alpha_s = (\lambda^{\max}_{\text{gen}, s})^2 \ \forall s$.

Although the convergence theorem only guarantees asymptotic convergence of the sequence $\{ \big{(}\bm{\lambda(t_n)}, \bm{p(t_n)} \big{)}\}_{t_n}$ to an optimal rate-price pair $\big{(} \hat{\bm{\lambda}}, \hat{\bm{p}}\big{)}$, in any realization of an EGS one expects that the convergence time $\Delta \tau$, the number of time slots that the RCP must run before convergence is attained, is finite. In addition, it is practically relevant to characterize the tightness of convergence $\delta$, or the maximum size of fluctuations about the optima. 

If an EGS is connected to $N$ nodes, there are $|S|_{\max} = \binom{N}{2}$ possible sessions. We assume that in a real network not all pairs of users require shared entanglement. In Fig. \ref{fig:convStudy} we numerically investigate the convergence time and tightness of convergence, $(\Delta \tau, \delta)$, for an EGS with $R=3$ resources and $p_{\text{gen}} = 0.05$ connected to $N=20, \ 50$ and $100$ users, where the number of sessions is restricted to $|S| = 0.1 \cdot |S|_{\max}$ by randomly sampling $10 \%$ of the possible sessions. In these simulations we set $x_{s} = 1 \ \forall s$, and average over 1000 independent runs of the simulation, each using the same set of sessions.

The reported convergence times $\Delta \tau$ are the number of time slots that occur before the sum of demand rates $\underset{s}{\Sigma} \lambda_s(t_n)$ first crosses the optimal value $\lambda_{\text{EGS}}$. Reporting of the tightness of convergence, $\delta$, is based on the maximum size of fluctuations of $\underset{s}{\Sigma} \lambda_s(t_n)$ about $\lambda_{\text{EGS}}$ following $\Delta \tau$. As the number of sessions hosted by an EGS increases, we observe a trade-off between $\Delta \tau$ and $\delta$. When the number of sessions is lower, $\Delta \tau$ is shorter but $\delta$ is larger. We have performed additional simulations which indicate that increasing the step size used in the RCP can be used to trade larger $\delta$ for somewhat shorter $\Delta \tau$.

If constraint changes occur slowly compared to $\Delta \tau$, Theorem \ref{thm:ConvergenceThm} implies that the RCP should re-establish convergence to a new optimal rate and price vector pair, $(\hat{\bm{\lambda}}, \hat{\bm{p}}) \mapsto (\hat{\bm{\lambda}}', \ \hat{\bm{p}}')$. In a real EGS system it is possible that the number of available resources will not be static in time, as resources may require periodic downtime for calibration. The effect of a change in the number of resources $R \mapsto R'$ changes the maximum service rate $\lambda_{\text{EGS}} = R \cdot p_{\text{gen}} \mapsto  \lambda'_{\text{EGS}} = R' \cdot p_{\text{gen}}$. To validate the robustness of the algorithm against such constraint changes we simulate EGS systems originally equipped with $R=3$ resource nodes, where after every $10,000$ time-slots one of the resources may either be taken offline for calibration or an offline resource may be returned to service. Fig. \ref{fig:avReqRates} demonstrates that the RCP successfully re-establishes convergence of $\underset{s}{\Sigma} \lambda_s(t_n)$ about $\lambda'_{\text{EGS}}$ following these constraint changes in an EGS system connected to $N = 50$ nodes, serving $|S| = 123$ communication sessions. 

In Fig. \ref{fig:avReqRates} we record the sequence of convergence times, $\{\Delta \tau \}$, after each constraint change as the first time-steps where $\underset{s}{\Sigma} \lambda_s (t_n)$ crosses $\lambda'_{\text{EGS}}$. To calculate the tightness of convergence, $\delta$, we first calculate the sequence of $\{ \delta' \}$, the size of the maximum fluctuations about $\lambda'_{\text{EGS}}$ following each $\Delta \tau'$ and set $\delta = \max(\{\delta' \})$. Notably, every subsequent $\Delta \tau ' < \Delta \tau$ and the achieved $\delta$ is equal to that observed when there are no changes to the constraint set in Fig. \ref{fig:convStudy} (middle plot, $\delta_2$) for an EGS with the same number of nodes, serving the same number of communication sessions. Additional simulations of EGS systems connected to various numbers of nodes ranging from $10$ to $100$, with random changes to the number of resources after every $10, 000$ time-steps, suggest that the data in Fig. \ref{fig:avReqRates} is representative. Specifically, in each case investigated the absolute relative difference,
\begin{align*}
    \frac{|\delta - \tilde{\delta}|}{\tilde{\delta}} < 1
\end{align*} between the achieved tightness of convergence when there are ($\delta$) and are not ($\tilde{\delta}$) changes to the constraints is less than $1$. 


\begin{figure*}[ht]
    \centering
    \includegraphics[scale=0.42]{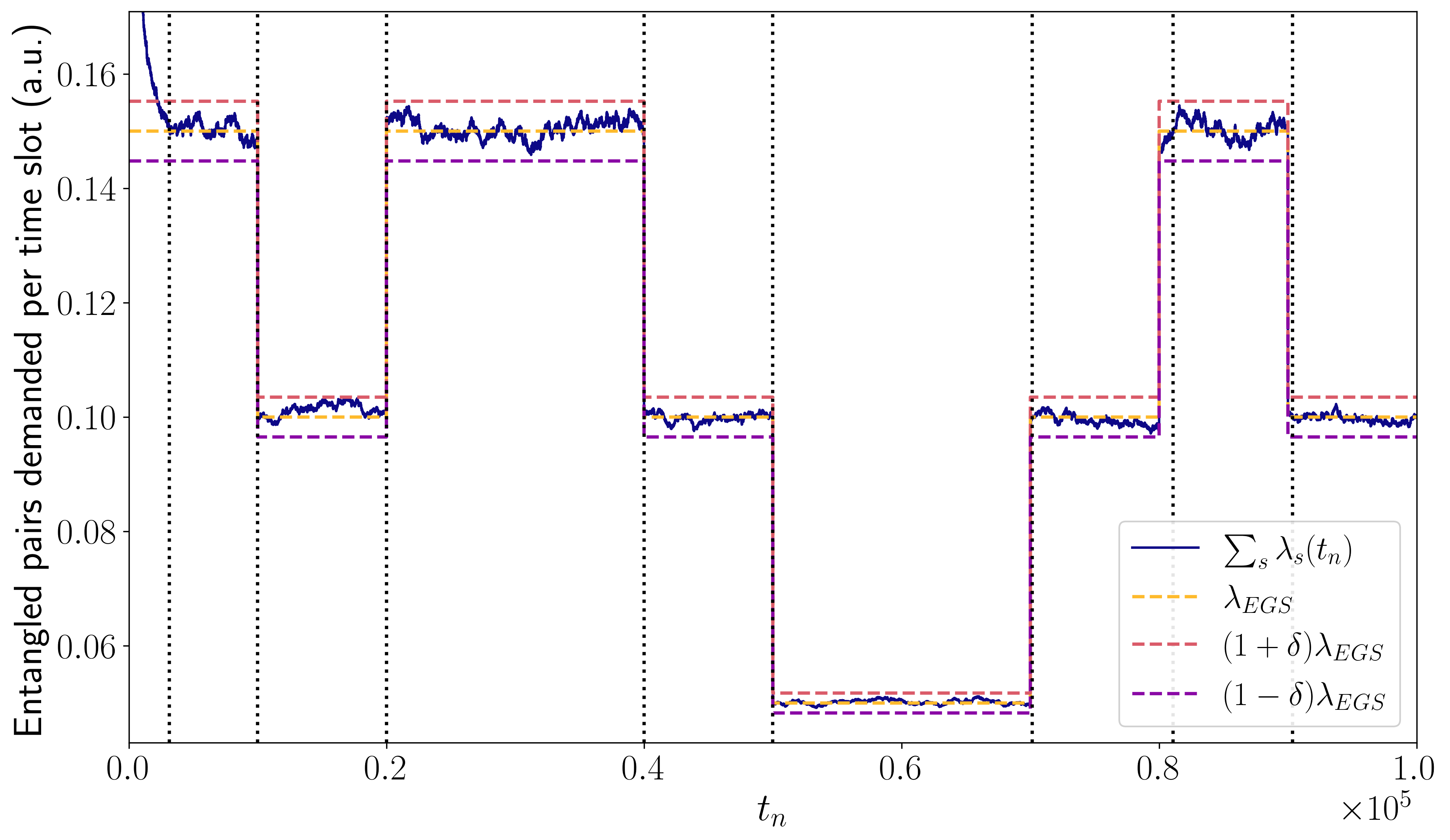}
    \caption{In response to changes in the number of resources available at the EGS ($R \rightarrow R'$), the RCP drives the sum of the demanded rates of entanglement generation across all communication sessions, $\underset{s}{\Sigma} \lambda_s (t_n)$, to  converge with respect to the sequence of time slots to the updated maximum average entanglement generation rate of the EGS, $\lambda_{\text{EGS}} = R' \cdot p_{\text{gen}}$. In simulation, an EGS connected to $N=50$ nodes, serving $|S|=123$ communication sessions, is originally equipped with $R=3$ resources. After every $10,000$ time-slots, one of the resources may either be taken offline for calibration or an offline node may be returned to service. Black dashed lines indicate the convergence, $\Delta \tau$ calculated for every $R'$ (initially $R$). We observe and overall tightness of convergence of $\delta=0.035$, identical to that observed in Fig. \ref{fig:avReqRates} for the EGS operated with fixed $R=3$ and with the same $N, \ |S|$.  Step-sizes $(\theta_c, \theta_u \ \forall u)$ were all $1/(10 \cdot \lambda_{\text{EGS}})$.}
    \label{fig:avReqRates}
\end{figure*}


The constraints $ \{ \lambda_u \}_u$ on the capabilities of nodes appear in (\ref{eq:userPriceUpdates}) and therefore affect both the prices calculated by the nodes and the rates set by communication sessions in (\ref{eq:formalAlgOptimalRates}). Since these constraints limit the total rate at which a node can submit demands summed across all of the communication sessions in which it participates, it is expected that uniform settings of $\{\lambda_u \}_u$ yield rate vectors under the RCP where $\{\lambda_s(t_n)\}_s$ are approximately uniform. In contrast, if the node constraints are non-uniform amongst the nodes, it is expected that the RCP yields rate vectors with larger differences between the rates set by each communication session. In Fig. \ref{fig:MaxMinDiffs} we investigate the effect of different settings for these constraints by plotting the difference between the average maximum $\underset{s}{\max} \{ \lambda_s(t_n)\}_s$ and minimum $\underset{s}{\min} \{ \lambda_s(t_n)\}_s$ communication session rates yielded by the RCP for two different settings of the constraints. In the first setting, node constraints are set uniformly as $\lambda_u = \big{(}(|S| -1)/2\big{)} \cdot p_{\text{gen}} \ \forall u $ so that in practice the algorithm functions as if the network node constraints have been removed. In the other setting there are three possible constraint values: a quarter of the nodes sampled at random have $\lambda_u = 1.5 \cdot p_{\text{gen}}$, half of the nodes have $\lambda_u = p_{\text{gen}}$, and a quarter of the  nodes have $\lambda_{u} = 0.5 \cdot p_{\text{gen}}$. Fig. \ref{fig:MaxMinDiffs} confirms that the difference between the average maximum rate and the average minimum rate requested by any session at time-step $t_n$ is one or more orders of magnitude larger when nodes are associated with the non-uniform constraint set. The uniform node constraint setting led to communication sessions updating their rates of demand submission to be nearly uniform across all communication sessions.

\begin{figure}[ht]
    \includegraphics[scale=0.23]{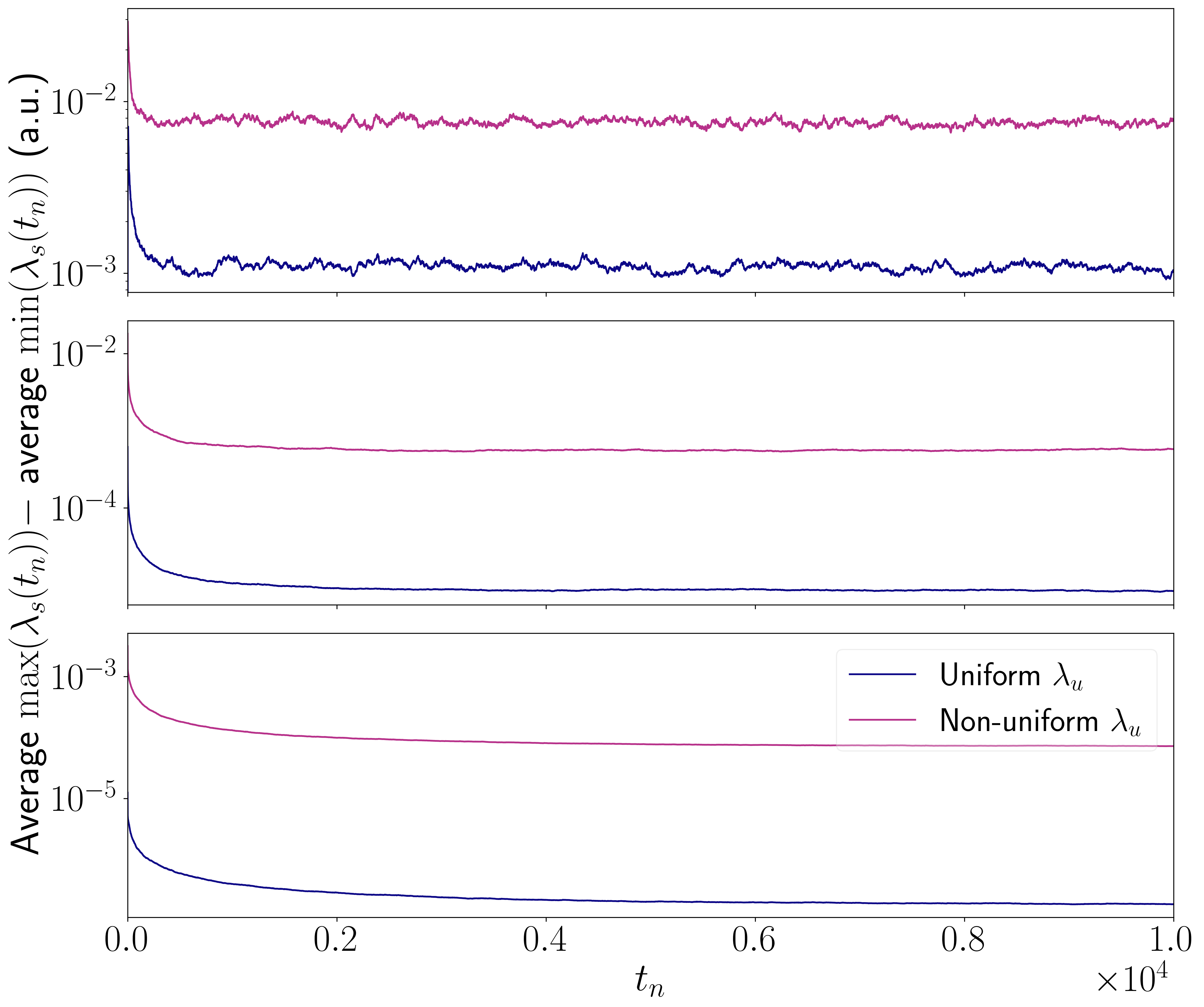}
    \caption{Differences between the average maximum rate and average minimum rate requested by any communication session in time-slot $t_n$, for an EGS connected to $N=20$ (top), $N=50$ (middle) and $N=100$ (bottom) nodes serving $|S| = 19, \ 123, \text{ and } 495$ communication sessions, respectively. As described in the main text, nodes are either associated with a uniform and effectively un-restricted set of capabilities or a non-uniform and more restricted set of capabilities. Step-sizes $(\theta_c, \theta_u \ \forall u)$ were all $1/(40 \cdot \lambda_{\text{EGS}})$.}
    \label{fig:MaxMinDiffs}
\end{figure}

%% file: LongOpenQuestions.tex
\section{Discussion}
We have presented the first control architecture for an EGS. The architecture is tailored to a simple system model. As a natural extension of this work, a refined version of the control architecture can be developed to suit a more versatile physical model. In the following discussion, we explore considerations for the development of a second generation control architecture.

In this work we assume a demand model in which user generated demands are fully parameterized by a desired rate of entanglement generation. Specifically, every communication session $s$ sets $\lambda_s(t_n)$, updated once per time-slot and specifies the constraint parameter $\lambda^{\min}_{s}$ which defines the minimum rate of entanglement generation the communication session must receive in order to enable some target application. While this model is mathematically simple, it may not fully address real application requirements on a physical quantum network. Real applications may require the simultaneous existence of a number of entangled pairs, each with some minimum fidelity and it is possible that applications need such packets of pairs to be supplied periodically over a longer application run-time. In the future, it may therefore be relevant to consider a demand model wherein communication sessions submit demands for packets of entanglement generation. A packet would be fully specified by the desired number of entangled pairs, a minimum fidelity for the pairs, some maximum window of time between the generation time of the first and last entangled pair of the request, and possibly some rate at which the demand with the preceding parameters should be repeatedly fulfilled.

The discussed model assumes that user controlled nodes can engage in multiple entanglement generation tasks in parallel. We do not impose restrictions on simultaneously scheduling communication sessions. Hence, it is possible for communication sessions $s$ and $s'$ with node $u \in s, s'$ to be scheduled simultaneously. Additionally, we consider the option of assigning multiple resource nodes to a single communication session in any time-slot. Therefore, we consider nodes with an unrestricted number of qubits and independent physical connections to the EGS. A subtlety we do not address here is that allocating multiple resources to a single communication session may require temporal multiplexing in the scheduling of individual entanglement generation attempts, especially when the multiple qubits of a single node are coupled to the physical connection via a single output. Furthermore, for nodes consisting of a single quantum processor, it may not be possible to calibrate the node to simultaneously engage in entanglement generation attempts with multiple partner nodes, even if the node has unlimited qubits. To capture this physical feature, it will be interesting to include the restriction of scheduling only non-overlapping communication sessions in the design of scheduling routines for future EGS control architectures.
 
The control architecture for an EGS relies on precise timing synchronization. Our model assumes that at both the control and physical layers, all communication sessions can adhere to the time slots defined by the EGS processor. Tight synchronization of timing is possible at the physical layer, which controls the quantum devices and coordinates the exact timing of entanglement generation attempts. However, tight timing synchronization of any type of classical communication may be a considerable challenge in any real world application. In particular, such coordination is a serious challenge if there are non-uniform communication times between any of the nodes and the EGS or between any of the node pairs. To reduce the timing requirements and possibly make the control architecture delineated here executable on a real-world system, it is possible to consider the processor of the EGS simulating the actions of the nodes. To do so, the processor would locally run the RCP and simulate the generation of demands originating from the user operated nodes by simply adding demands to the queues based on the rates output by the RCP. Such an approach trades the difficulty of timing synchronization for the requirement of increased power of the classical processor at the EGS. To reduce the need for timing synchronization, a second generation architecture may be designed which does not rely on fixed, centrally defined time slots.

%% file: Proofs.tex
\subsection{Proofs}
\label{sec:Proofs}
\subsubsection{Outline of goals to prove}
In this section we will prove two theorems to establish the results quoted in the main body of the article. The results are as follows:
\begin{enumerate}
    \item The capacity region of the EGS is the set of demand arrival rate vectors fully contained in the set $\mathcal{C}$
    (\ref{eq:CapReg}) and maximum weight scheduling (Definition \ref{def:MaxWeight}) supports any rate vector from within $\mathcal{C}$ (Theorem \ref{thm:CapReg}). To establish the capacity region, we first prove a proposition stating that any rate vector $\bm{\lambda} \ \cancel{\in} \ \mathcal{C}$ necessarily results in divergent queues. We then prove a second proposition establishing at once that any rate vector $\bm{\lambda} \in \text{Int}\mathcal{C}$ is supportable under some scheduling algorithm and that maximum weight scheduling is such a scheduling algorithm. Therefore, we also demonstrate that maximum weight scheduling is throughput optimal. 
    \item The RCP, Algorithm \ref{Box:rateControlProt}, results in the calculation of a sequence of rate and price vector pairs  $\big{(}\bm{\lambda}(t_n) , \bm{p}(t_n)\big{)}$ which converge to optimal solutions $\big{(}\hat{\bm{\lambda}}, \ \hat{\bm{p}}\big{)}$ of the primal and dual problems, defined in Section \ref{sec:Algorithm} (Theorem \ref{thm:ConvergenceThm}). 
\end{enumerate}

\subsubsection{Proof of Theorem \ref{thm:CapReg}}
\label{subsec:proofCapReg}
 First it is to be shown that no rate vector $\bm{\lambda} \ \cancel{\in} \ \mathcal{C}$ of an EGS with $R$ resources is supportable under any scheduling algorithm. 
\begin{proposition}
\label{prop:NotSupportable}
If $\bm{\lambda}\ \cancel{\in} \ \mathcal{C}$, no scheduling algorithm can support $\bm{\lambda}$. 
\end{proposition}

\begin{proof}

There are three cases where $\bm{\lambda}\ \cancel{\in} \ \mathcal{C}$,
\begin{enumerate}
    \item $\underset{s}{\sum} \lambda_s > R \cdot p_{\text{gen}}$, or
    \item $\lambda_{s^{*}} > x_{s^{*}} \cdot p_{\text{gen}} \text{ for some } s^{*} \in  S$.
    \item $\bm{\lambda}$ is not non-negative ( $\exists \ \lambda_{s^{*}} < 0$ for at least some $s^{*} \in S$).  
\end{enumerate} 
In the third case, the node pair corresponding to session $s^{*}$ has set a non-physical rate and the rate must be changed. The proof for case $(2)$ is very similar to case $(1)$ and equations from the first case are re-used or modified to complete the proof of case $(2)$. The main strategy of the proof relies on Definition \ref{def:supportable}; a rate vector $\bm{\lambda}\ \cancel{\in} \ \mathcal{C}$ is not supportable if $\bm{\lambda}$ causes the queue lengths at the EGS processor to diverge with probability 1, regardless of scheduling algorithm. To prove the proposition in each case, it serves to calculate the total queue length. 

\noindent
\textit{Proposition \ref{prop:NotSupportable} (1):} Suppose $\underset{s}{\sum} \lambda_s > R \cdot p_{\text{gen}}.$ Then, $\exists \ \epsilon > 0$ such that
    \begin{align}
\label{eq:constrainArrivals}
\sum_{s} \lambda_s \geq R \cdot p_{\text{gen}} + \epsilon.\end{align}

Assume that the initial length of each queue is finite. The sum of queue lengths at time step $t_{n+1}$,  $\underset{s}{\sum} q_s(t_{n+1})$ is, 
\begin{align}
    \sum_{s}  q_s(t_{n+1}) &= \sum_{s}  \big{[} q_s(t_n) + a_s(t_n) -g_s(t_n)\big{]}^{+} \nonumber \\
    & \geq \sum_{s}  \big{(}q_s(t_n) + a_s(t_n) - g_s(t_n) \big{)} \nonumber \\
    & \label{eq:QueueTelescopic} \geq \sum_{s}  \Big{(}q_s(t_1) + \overset{t_n}{\underset{t_i=t_1}{\sum}} \big{(}a_s(t_i) - g_s(t_i) \big{)} \Big{)}
\end{align}
\noindent
where $a_s(t_i)$ is the integer number of demands submitted by communication session $s$ at time step $t_i$ and $g_{s}(t_i)$ is the integer number of successfully generated entangled pairs between the nodes corresponding to communication session $s$ in time step $t_i$. The final inequality in (\ref{eq:QueueTelescopic}) follows from the previous inequality by repeated application of (\ref{eq:queuesUpdate}). By the strong law of large numbers,
\begin{align}
\label{eq:singleSLLN}
    \underset{n \rightarrow \infty}{\text{lim}}  \ \frac{1}{t_n} \ \overset{t_n}{\underset{t_i=t_1}{\sum}} \ a_s(t_i) =  \lambda_s \ \forall \ s \in S \text{, with probability 1}.
\end{align}
Recall that the number of successfully generated entangled pairs between the nodes corresponding to communication session $s$ at time $t_i$ is a sample from a binomial random process where the number of trials is set by $M_{s}(t_i)$ and the trial success probability is $p_{\text{gen}}$, 
\begin{align*}
    g_{s}(t_i) \sim \text{Bin} \big{(}M_{s}(t_i), \ p_{\text{gen}}\big{)}.
\end{align*}

By the strong law of large numbers, 
\begin{align}
\label{eq:ExpectedPairGen}
\underset{n \rightarrow \infty}{\text{lim}}  \ \frac{1}{t_n} \ \overset{t_n}{\underset{t_i=t_1}{\sum}} \  g_s(t_i) = M_{s}(t_n) \cdot p_{\text{gen}} \text{, with probability 1}.
\end{align}

Since each feasible schedule satisfies $\underset{s}{\sum} M_s(t_i) \leq R$, it follows from (\ref{eq:ExpectedPairGen}) that 
\begin{align}
\label{eq:scheduleSumIneq}
\underset{n \rightarrow \infty}{\text{lim}}  \ \frac{1}{t_n} \ \overset{t_n}{\underset{t_i=t_1}{\sum}} \  \sum_{s}
\ g_s(t_i) &= \sum_{s} \ \underset{n \rightarrow \infty}{\text{lim}}  \ \frac{1}{t_n} \ \overset{t_n}{\underset{t_i=t_1}{\sum}}
\ g_s(t_i) \nonumber \\
&\leq R \cdot p_{\text{gen}},
\end{align}
where we use the distribution property of limits, which is possible because the individual limits (\ref{eq:ExpectedPairGen}) exist. Finally, by assumption (\ref{eq:constrainArrivals}) and (\ref{eq:QueueTelescopic}, \ref{eq:singleSLLN}) and (\ref{eq:scheduleSumIneq}), \small
\begin{align}
    \underset{n \rightarrow \infty}{\text{lim}} &\ \frac{1}{t_n} \ \sum_{s}   q_s(t_{n+1}) \nonumber \\ 
    &\geq \underset{n \rightarrow \infty}{\text{lim}}  \ \frac{1}{t_n} \ \sum_{s} q_s(t_1) \nonumber \\ 
    & \phantom{x} + \underset{n \rightarrow \infty}{\text{lim}}  \ \frac{1}{t_n} \ \overset{t_n}{\underset{t_i=t_1}{\sum}} \big{(} \sum_{s} a_s(t_i) -  \sum_{s}  M_s(t_i)\big{)} \nonumber \\
    &\geq \sum_{s} \ \lambda_s - R \cdot p_{\text{gen}} \nonumber \\
    & \geq R \cdot p_{\text{gen}} + \epsilon - R \cdot p_{\text{gen}} \nonumber \\
    & \geq \epsilon. 
\end{align}
\normalsize
Therefore, with probability 1, $\underset{s}{\sum} q_s(t_n) \rightarrow \infty$ as $n \rightarrow{\infty}, $ so $\bm{\lambda}$ is not supportable, regardless of scheduling algorithm.

\noindent
\textit{Proposition \ref{prop:NotSupportable} (2):} Suppose that $\lambda_{s^*} > x_{s^{*}} \cdot p_{\text{gen}}$ for some $s^* \in S$. Then, $\exists \ \epsilon > 0$ such that,
\begin{align}
\label{eq:p2Violation}
    \lambda_{s^*} \geq x_{s^{*}} \cdot p_{\text{gen}} + \epsilon.
\end{align}

In this case, we show that $\bm{\lambda}$ is not supportable by proving that the queue $q_{s^{*}}(t_i)$ of demands associated with communication session $s^{*}$ diverges for large $t_i$. Recall (\ref{eq:ExpectedPairGen}) and note $M_{s}(t_i) \leq x_s \ \forall \ s, \ \forall \ t_i$. This inequality describes that a maximum of $x_s$ heralding stations can be allocated any communication session $s$ in $t_i$. With this restriction, (\ref{eq:ExpectedPairGen}) becomes,
\begin{align}
\label{eq:singlePairGenLimit}
\underset{n \rightarrow \infty}{\text{lim}}  \ \frac{1}{t_n} \ \overset{t_n}{\underset{t_i=t_1}{\sum}} \  g_s(t_i) \leq x_s \cdot p_{\text{gen}} \ \forall s. 
\end{align}

Combining assumption (\ref{eq:p2Violation}) using (\ref{eq:singleSLLN}), (\ref{eq:singlePairGenLimit}), and making repeated use of (\ref{eq:queuesUpdate}),
\begin{align}
\underset{n \rightarrow \infty}{\text{lim}} \ \frac{1}{t_n} \ & q_{s^{*}}(t_{n + 1}) \nonumber \\ 
&\geq \underset{n \rightarrow \infty}{\text{lim}}  \ \frac{1}{t_n} \  q_{s^{*}}(t_1) \nonumber \\ 
& \phantom{x} +  \underset{n \rightarrow \infty}{\text{lim}} \
\frac{1}{t_n} \ \overset{t_n}{\underset{t_i=t_1}{\sum}} \big{(} a_{s^{*}}(t_i) - g_{s^{*}}(t_i) \big{)} \nonumber \\
& \geq \lambda_{s^{*}} - x_{s^{*}} \cdot p_{\text{gen}} \nonumber \\
& \geq x_{s^{*}} \cdot p_{\text{gen}} + \epsilon - x_{s^{*}} \cdot p_{\text{gen}} \nonumber \\
\label{eq:concludeSingleOverloadedQueue} & \geq \epsilon.
\end{align}

Therefore, with probability $1$, $q_{s^{*}}(t_{n+1}) \rightarrow \infty$ as $n \rightarrow \infty$. Hence $\bm{\lambda}$ is not supportable. 
\end{proof}

Proposition \ref{prop:NotSupportable} proved that rate vectors  $\bm{\lambda} \ \cancel{\in} \ \mathcal{C}$ are not in the capacity region of the EGS. To finish proving $\mathcal{C}$ is the capacity region of the EGS (Theorem \ref{thm:CapReg}), it is necessary to prove that any rate vector $\bm{\lambda}\in \mathcal{C}$ is supportable under some scheduling algorithm. To do so, we prove that the specific scheduling algorithm of Maximum Weight Scheduling (Definition \ref{def:MaxWeight}) supports all arrival rate vectors fully contained in $\mathcal{C}$. 

\begin{proposition}
\label{prop:MaxWeightThroughputOptimal}
Maximum Weight scheduling can support any arrival rate vector $\bm{\lambda}$ for which $\exists \epsilon > 0$ such that $(1+\epsilon) \bm{\lambda}\in \mathcal{C}$.
\end{proposition}

Modelling a queue vector as a Markov chain is a standard tool in queuing theory \cite{SrikantYing}. This approach makes it possible to take advantage of the many strong analytic results on the behaviour of Markov chains, which can then be used to make statements about the queue vector. The vector $\bm{q}(t_n) = \big{(}q_s(t_n\big{)} \ \forall \ s)$ of queued demands from each communication session maintained in the processor at $t_n$ can be modelled as a Markov chain, with transitions given by (\ref{eq:queuesUpdate}). An irreducible Markov chain has the property that any state $i$ of the chain is reachable from any other state $j$. A positive recurrent Markov chain has the property that from any state $i$, the expectation value of the time it will take to re-visit any other state $j$ is finite. A queue vector, with specified dynamics, that can be modelled as an irreducible Markov chain with the property of positive recurrence will not diverge (i.e. is guaranteed to remain a finite queue) \cite{SrikantYing}. The dynamics of such a queue vector are fixed by the arrival rate vector and the scheduling routine, therefore if a queue vector can be modelled as a positive recurrent Markov chain, the arrival rate vector is supportable by the scheduling routine. To prove Proposition \ref{prop:MaxWeightThroughputOptimal} we demonstrate that the queue vector is an irreducible Markov chain and use the \href{thm:Foster-Lyapunov}{Foster-Lyapunov Theorem} to prove that whenever $\bm{\lambda}$ lies strictly within $\mathcal{C}$ the Markov chain is also positive recurrent. An equivalent statement is that all rate vectors lying strictly within $\mathcal{C}$ are supportable by some scheduling algorithm.
\begin{theorem}[Foster-Lyapunov Theorem \cite{SrikantYing}]
\label{thm:Foster-Lyapunov}
Let $\{X_k\}$ be an irreducible Markov chain with a state space $\mathcal{S}$. Suppose that there exists a function $V \ : \ \mathcal{S} \rightarrow \mathbb{R}^{+}$ and a finite set $\mathcal{B} \subseteq \mathcal{S}$ satisfying the following conditions:
\begin{enumerate}
    \item $\mathbb{E}[V(X_{k+1}) - V(X_k) | X_k = x]  \leq - \epsilon$ if $x \in \mathcal{B}^c$, for some $\epsilon > 0$, and
    \item $\mathbb{E}[V(X_{k+1}) - V(X_k) | X_k = x]  \leq A$ if $x \in \mathcal{B}$, for some $A < \infty$.  
\end{enumerate}
Then the Markov chain $\{ X_k\}$ is positive recurrent.
\end{theorem}

\begin{proof}[Proof of Proposition \ref{prop:MaxWeightThroughputOptimal}]
First we establish that the queue vector, $\bm{q}(t_i) \ \forall \ t_i$ is an irreducible Markov chain. The queue vector, $\bm{q}(t_i)$ is a Markov chain with state space 
\begin{align*}
    \mathcal{S} = \{ \bm{q} : \ & \bm{q} \text{ is reachable from } \bm{0} \\
    &  \text{ under the given scheduling algorithm} \}.
\end{align*}

Assume that $\bm{q}(t_1)$ is finite and $\bm{q}(t_1) \in \mathcal{S}$. It follows from the definition of $\mathcal{S}$ that  $\bm{q}(t_i) \in \mathcal{S} \ \forall t_i$ if $\bm{q}(t_1) \in \mathcal{S}$. Irreducibility of $\bm{q}(t_i) \ \forall t_i$ requires that any state $\bm{q}(t_j)$ is reachable from any other state $\bm{q}(t_i)$. By the definition of the state space $\mathcal{S}$, it suffices to demonstrate that from $\bm{q}(t_i)$, the Markov chain can always return to $\bm{0}$. Under Maximum Weight scheduling (Definition \ref{def:MaxWeight}), the processor always serves $k(t_i)$ demands per time-slot, where 
$$k (t_i)  = \text{max} \{k \ : \ k \leq R\text{ and } k \leq \underset{s}{\Sigma}  \ \min \big{(}|q_s(t_i)|, \ x_{s} \big{)}\},$$
where $|q_{s}(t_i)|$ is the number of demands in the queue for session $s$ in time-slot $t_i$ and $x_{s}$ is the maximum number of resource modules that can be allocated communication session $s$ per time-slot. Hence when $\bm{q}(t_i)$ is non-zero, at least one demand and up to $R$ demands are served per time-slot. Therefore, from any $\bm{q}(t_i) \in \mathcal{S}$, $\bm{q}(t_{i + l}) = \bm{0}$ is reachable from $\bm{q}(t_i)$ in  $l \in \{ \ceil{\frac{|\bm{q}(t_i)}{R}},  \ceil{\frac{|\bm{q}(t_i)}{R}} + 1, \cdots, \ |\bm{q}(t_i)|\}$ time steps, where $|\bm{q}(t_i)| := \underset{s}{\Sigma} \ |q_s(t_i)|$. Since any other $\bm{q}(t_j) \in \mathcal{S}$ is then reachable from $\bm{0}$, it follows that $\bm{q}(t_i)$ is irreducible. To prove that $\bm{\lambda} $ is supportable, it suffices to demonstrate that $\bm{q}(t_i)$ is positive recurrent. 

Define the Lyapunov function 
\begin{align}
    L\big{(} \bm{q}(t_i)\big{)} = \frac{1}{2} \sum_{s}q^{2}_s(t_i).
\end{align}

To apply the the Foster-Lyapunov theorem (\ref{thm:Foster-Lyapunov}), the key quantity is the drift of $L\big{(}\bm{q}(t_i)\big{)}$. Using the queue update dynamics (\ref{eq:queuesUpdate}), the drift can be expanded as
\begin{align}
\label{eq:LyapunovDrift}
    L\big{(}&\bm{q}(t_{i + 1})\big{)} - L\big{(} \bm{q}(t_i) \big{)} \nonumber \\
    & = \frac{1}{2}\sum_{s} \Big{(} \big{[} q_s(t_i) + a_{s}(t_i) - g_{s}(t_i) \big{]}^{+} \Big{)}^{2} - \frac{1}{2}\sum_{s} q^{2}_{s}(t_i) \nonumber \\
    & \leq \frac{1}{2}\sum_s \big{(} q_{s}(t_i) + a_s(t_i) - g_{s}(t_i) \big{)}^2 - \frac{1}{2}\sum_{s} q^2_{s}(t_i) \nonumber \\
    &= \frac{1}{2}\sum_{s} \big{(} a_{s}(t_i) - g_{s}(t)\big{)}^2 \nonumber \\
    & \phantom{x} + \sum_{s} q_{s}(t_i) \big{(} a_s(t_i) - g_{s}(t_i)\big{)}.
\end{align}
Taking the conditional expectation of the Lyapunov drift with respect to the randomness of arrivals and the probabilistic success of scheduled demands,
\begin{align}
\label{eq:CondExDrift1}
\mathbb{E} \big{[} & L \big{(} \bm{q}(t_{i + 1})\big{)} - L \big{(} \bm{q}(t_i)\big{)} \ | \ \bm{q}(t_i) = \tilde{\bm{q}} \big{]} \nonumber \\
& \leq \frac{1}{2} \sum_{s} \mathbb{E}\big{[}  \big{(}a_s(t_i) - g_s(t_i)\big{)}^2  | \ \bm{q}(t_i) =\tilde{\bm{q}}\big{]}  \nonumber\\ 
& + \sum_{s} \mathbb{E} \big{[}  q_s(t_i) \big{(} a_s(t_i) - g_s(t_i)\big{)} \ | \ \bm{q}(t_i) = \tilde{\bm{q}}\big{]},
\end{align}
where $\tilde{\bm{q}} \in \mathcal{S}$ is a particular queue vector. 

Using $\big{(} a_s - g_s\big{)}^2 \leq a^{2}_{s} + {g_{s}}^{2}$ and the linearity of expectation, the first term of the conditional expectation can be re-written, 
\begin{align}
\label{eq:squaredExpectations}
    \mathbb{E}\big{[}& \sum_{s} \big{(}a_{s}(t_i) - g_{s}(t_i) \big{)}^2 \ | \ \bm{q}(t_i) = \tilde{\bm{q}} \big{]} \leq \nonumber \\ 
    & \sum_{s} \mathbb{E} \big{[} a^{2}_{s}(t_i) \ | \ \bm{q}(t_i) = \tilde{\bm{q}} \big{]} + \sum_{s} \mathbb{E} \big{[} {g_{s}}^{2}(t_i) \ | \ \bm{q}(t_i) = \tilde{\bm{q}}\big{]}
\end{align}

Recall that $g_{s}(t_i) \leq M_{s}(t_i) \leq x_s \ \forall \ s, \ \forall \ t_i$. Hence, 
\begin{align}
\label{eq:expectationSquaredPairProduction}
    \underset{s}{\sum} \mathbb{E} \big{[} {g_{s}}^{2} (t) \ | \ \bm{q}(t) = \tilde{\bm{q}}\big{]} \leq \sum_{s} x_s^2. 
\end{align}

Define the variance in the arrivals to the queue of session $s$, $\sigma^{2}_{s} := \text{Var}[a_{s}(t_i)]$. Then, noting that the arrivals are independent of the state of the queues, using the definition of variance and $\mathbb{E}[a_s(t_i)] = \lambda_s$,

\begin{align}
\label{eq:arrivalVar}
\mathbb{E}[a^2_{s}(t_i) \ | \ \bm{q}(t_i) = \tilde{\bm{q}}] = \mathbb{E}[a^2_{s}(t_i)] = \sigma^2_{s} + \lambda^2_{s}
\end{align}

Together (\ref{eq:expectationSquaredPairProduction}) and (\ref{eq:arrivalVar}) bound the first term of (\ref{eq:CondExDrift1}),
\begin{align*}
    \frac{1}{2} \sum_{s} & \mathbb{E}\big{[} \big{(}a_s(t_i) - g_s(t_i)\big{)}^2  | \ \bm{q}(t_i) =  \tilde{\bm{q}}\big{]} \\ 
    & \leq \frac{1}{2} \sum_{s} \ \big{(} \sigma^{2}_{s} + \lambda^{2}_{s} +  x_s^2 \big{)} =: B 
\end{align*}

Then (\ref{eq:CondExDrift1}) is,
\begin{align}
    \mathbb{E} \big{[}& L \big{(} \bm{q}(t_{i + 1})\big{)} - L \big{(} \bm{q}(t_i)\big{)} \ | \ \bm{q}(t_i) = \tilde{\bm{q}} \big{]}  \nonumber \\ 
    &\leq B +  \ \sum_{s} \ \mathbb{E} \big{[} q_{s}(t_i) \big{(} a_s(t_i) - g_{s}(t_i) \big{)} \ | \ \bm{q}(t_i) = \tilde{\bm{q}} \big{]} \nonumber \\
    & = B + \ \sum_{s} \ \tilde{q}_s \big{(} \lambda_s - \mathbb{E} \big{[}  g_{s}(t_i) \ | \ \bm{q}(t_i) = \tilde{\bm{q}}\big{]} \big{)}.
\end{align}

Recall that the conditional expectation of the Lyapunov drift is taken with respect to the randomness of the arrival processes as well as the success of scheduled demands. The schedule selected for a given time-slot depends on the queues, but the success of any scheduled demand does not. The conditional expectation of pair production for communication session $s$ can be re-written as,
\begin{align}
\label{eq:ExPairProduction}
    \mathbb{E} [g_s(t_i) \ | \ \bm{q}(t_i) = \tilde{\bm{q}}] = p_{\text{gen}} \cdot \mathbb{E}[M_s(t_i) \ | \ \bm{q}(t_i) =\tilde{\bm{q}}]. 
\end{align}

Recall that $\bm{M}$ denotes the schedule decided under the maximum weight scheduling policy, \ref{def:MaxWeight}. Allow $\tilde{\bm{M}}$ to denote a schedule that is decided by any other scheduling policy. It follows from Definition \ref{def:MaxWeight} that, \small
\begin{align}
\label{eq:schedCompare}
    \underset{s}{\sum} \ \tilde{q}_s  \cdot \mathbb{E}[M_s(t_i) \ | \ \bm{q}(t_i) = \tilde{\bm{q}}] \geq \sum_{s} \ \tilde{q}_s \cdot \mathbb{E} [\tilde{M}_{s}(t_i) \ | \ \bm{q}(t_i) = \tilde{\bm{q}}]. 
\end{align} \normalsize
Consider a scheduling policy $\tilde{\bm{M}}$ which schedules each session at a rate of $\frac{\lambda_s + \epsilon}{p_{\text{gen}}}$ (this is possible since, by assumption, $(1 + \epsilon ) \bm{\lambda} \in \mathcal{C}$). Such a scheduling policy is aware of the demand arrival rates to each queue but is not demand based (i.e. it does not use queue information in deciding the schedule). Hence,

\begin{align}
\label{eq:altSched}
    \sum_s \ \tilde{q}_s \cdot \mathbb{E} [\tilde{M}_{s}(t_i) \ | \ \bm{q}(t_i) = \tilde{\bm{q}}] &= \sum_{s} \ \tilde{q}_s \cdot \mathbb{E} [\tilde{M}_{s}(t_i)] \nonumber \\ 
    & = \sum_{s} \ \tilde{q}_s \bigg{(}\frac{\lambda_s + \epsilon}{p_{\text{gen}}}\bigg{)}.
\end{align}
Combining (\ref{eq:ExPairProduction}), (\ref{eq:schedCompare}) and (\ref{eq:altSched}), the conditional expectation of the Lyapunov drift is bounded by, \small
\begin{align}
    \mathbb{E} \big{[}& L \big{(} \bm{q}(t_{i + 1})\big{)} - L \big{(} \bm{q}(t_i)\big{)} \ | \ \bm{q}(t_i) = \tilde{\bm{q}}\big{]}  \nonumber \\ 
    & \leq B + \ \sum_{s} \ \tilde{q}_s \lambda_s -  \sum_{s} \ \tilde{q}_s \cdot p_{\text{gen}} \cdot \mathbb{E} \big{[}  M_{s}(t_i) \ | \ \bm{q}(t_i) = \tilde{\bm{q}}\big{]} \big{)} \nonumber \\
    & \leq B + \sum_{s} \ \tilde{q}_s \lambda_s - p_{\text{gen}} \cdot \sum_{s} \ \tilde{q}_s \cdot \mathbb{E} [\tilde{M}_{s}(t_i)] \nonumber \\
    & \label{eq:ExCondDriftFinal} = B - \epsilon \cdot \sum_{s} \ \tilde{q}_s. 
\end{align}
\normalsize
Application of the Foster-Lyapunov theorem completes the proof.

\end{proof}

\subsubsection{Proof of Theorem \ref{thm:ConvergenceThm}}
\label{subsec:ConvergenceProof}
The proof of this theorem is closely inspired by the proof of an analogous theorem in \cite{FlowControlI}. To begin, we establish basic properties of the Dual function which follow from assumption \href{asA1}{A1}. 
\begin{lemma}
\label{lemma:dualconvex}
Under assumption \href{asA1}{A1} the dual objective function $D(\bm{p})$ (\ref{eq:dualObjective}) is convex, lower bounded, and continuously differentiable. 
\end{lemma}

For each session $s \in S$, and price vector $\bm{p} \geq 0$, define the quantity $\beta_{s}(\bm{p}): \mathbb{R}^{1 + N} \mapsto \mathbb{R}^{+}$ as follows,
\begin{align}
\label{eq:defBetaS}
    \beta_{s}(\bm{p}) : = \begin{cases}
    \frac{1}{-f_s^{''}\big{(} \lambda^{*}_s(\bm{p})\big{)}},  \text{ if } f'_{s}(\lambda^{\max}_{\text{gen},s}) \leq p_s \leq f'_{s}(\lambda^{\min}_{s}) \\
    0, \text{ otherwise,}
    \end{cases}
\end{align}
where $p_s := p_c + \underset{u \in s}{\sum} p_u$ and $\lambda^{*}_{s}(\bm{p})$ is the unique maximizer of (\ref{eq:dualObjective}).

For any price vector $\bm{p} \geq \bm{0}$ define the matrix $B(\bm{p}) = diag \big{(}
 \beta_s(\bm{p}), \ s \in S\big{)}$ to be the $|S| \times |S|$ matrix with diagonal elements $\beta_s(\bm{p})$. Note that from assumption \href{asA2}{A2}, for all $\bm{p} \geq \bm{0}$,
 \begin{align}
     0 \leq \beta_s(\bm{p}) \leq \alpha_s < \infty. 
 \end{align}
 
 Define the user-session mapping matrix $R$ to be the $N \times |S|$ matrix whose $(u, \ s)$-th entry is given by, 
 \begin{align}
     R_{u}^{\ s} = \begin{cases}
     1, \text{ if } u \in s \text{ or equivalently } s \in S(u) \\
     0, \text{ otherwise.}
     \end{cases}
 \end{align}
 
The augmented session mapping matrix $\tilde{R}$ is the $\big{(}1 + N\big{)} \times |S|$ matrix whose $(r, \ s)$-th entry is,
 \begin{align}
    \label{eq:augR}
    \tilde{R}_{r}^{\ s} = \begin{cases}
    1, \text{ if } r=1 \\
    R_{r-1}^{ \ \ s}, \ r \neq 1.
    \end{cases}
 \end{align}

 \begin{lemma}
 \label{lemma:dualHessian}
Under assumption \href{asA1}{A1}, where it exists, the Hessian of the dual function $D$ is given by
\begin{align}
\label{eq:dualHessian}
\grad^2 D(\bm{p}) = \tilde{R} B(\bm{p}) \tilde{R}^{\text{T}}.
\end{align}
 \end{lemma}
 
\begin{proof}
Let $\grad_p \lambda^{*}$ denote the $|S| \times (1+N)$ Jacobian matrix whose $(s, \ r)$-th element is $\big{(} \partial \lambda^{*}_s / \partial p_r\big{)} (\bm{p}), \ r \in (c, \ u \ \forall u ) $. As a consequence of the Inverse function theorem \cite{Rudin} and (\ref{eq:formalOptimalRates}), when it exists, 
\begin{align}
\label{eq:piecesgradLambda}
\frac{\partial \lambda^{*}_s}{\partial p_r} = \begin{cases}
\frac{\tilde{R}_{r}^{\ s}}{f^{''}_{s}\big{(}\lambda^{*}_{s}(\bm{p})\big{)}}, \text{ if } f^{'}_{s}(\lambda^{\max}_{\text{gen}, s}) < p_{s} < f^{'}_{s}(\lambda^{\min}_{s}); \\
0 , \text{ otherwise; }
\end{cases}
\end{align}
where $r \in (c, \ u \ \forall u ).$ Using (\ref{eq:defBetaS}) we can write,
\begin{align}
\label{eq:gradLambda} 
\grad_p \lambda^{*} = -B(\bm{p}) \tilde{R}^{\text{T}}. 
\end{align}
From (\ref{eq:dualGradC}) and (\ref{eq:dualGradU}), $\grad D(\bm{p})  = c - \tilde{R} \lambda$, where $c : = (\lambda_{\text{EGS}}, \ \overline{\lambda}_u \ \forall u)$, therefore,
\begin{align*}
\grad^2 D(\bm{p}) = - \tilde{R} \grad_{p} \bm{\lambda}= \tilde{R} B(\bm{p}) \tilde{R}^{\text{T}}. 
\end{align*}
\end{proof} 

\begin{lemma}
\label{lemma:dualLipschitz}
Under assumptions A1 and A2, the gradient of the dual function $\grad D(\bm{p})$ (\ref{eq:dualObjective}, \ref{eq:dualGradC}, \ref{eq:dualGradU}) is Lipschitz continuous with Lipschitz constant $L = \underset{s \in S}{\max} \ \beta_{s}(p) \cdot  \ |S|$. 
\end{lemma}

We use the following theorem to prove Lemma \ref{lemma:dualLipschitz},
\begin{theorem}[Rudin, 9.19 \cite{Rudin}]
\label{thm:Rudin9.19}
Suppose $\mathbf{f}$ maps a convex open set $E \subset \mathbb{R}^n$ into $\mathbb{R}^m$, $\mathbf{f}$ is differentiable in $E$, and there is a real number $M$ such that
\begin{align*}
    ||\mathbf{f}'(\mathbf{x})|| \leq M,
\end{align*}
for every $\mathbf{x} \in E$. Then,
\begin{align*}
    |\mathbf{f}(\mathbf{b}) - \mathbf{f}(\mathbf{a})| \leq M |\mathbf{b} - \mathbf{a}|, 
\end{align*}
for all $\mathbf{a} \in E, \ \mathbf{b} \in E$. 
\end{theorem}
\begin{proof}[Proof of Lemma \ref{lemma:dualLipschitz}]
From Lemma \ref{lemma:dualHessian}, the Hessian of the dual function is the $\big{(} 1 + N \big{)} \times \big{(} 1 + N \big{)}$ matrix $\grad^2 D(p) = \tilde{R}B(p)\tilde{R}^{\text{T}}$. It is simple to explicitly determine the $(r, r')$-th entry of $\grad^2 D(p)$. By matrix multiplication, $B(p)\tilde{R}^{\text{T}}$ is the $|S| \times \big{(} 1 + N \big{)}$ matrix whose $(s, r)$-th entry is, \small
\begin{align}
    \big{(} B(p)\tilde{R}^{\text{T}}\big{)}_{s}^{\ r} = \begin{cases}
    \beta_s(p), \text{ if } r=1 \text{ or } \big{(}r > 1 \text{ and } s \in S(r-1) \big{)} \\
    0, \text{ otherwise}.
    \end{cases}
\end{align} \normalsize

By matrix multiplication we calculate the $(r, \ r')$-th entry of $\grad^2 D(p)$ as,
\begin{align}
\big{(}\grad^2 D(p) \big{)}_{r}^{\ r'} &= \sum_{s} R_{r}^{\ s} \big{(} B(p) R^{\text{T}}\big{)}_{s}^{\ r'} \nonumber \\
&= \begin{cases} 
\sum_s  \ \beta_s, \ r=r' =1 \\
\underset{s \in S(r' -1)}{\sum} \beta_s, \ r=1 \text{ and } r'>1 \\
\underset{s \in S(r -1)}{\sum} \beta_s, \ r>1 \text{ and } r'=1 \\
\underset{s \in S(r-1) \cap S(r' -1)}{\sum} \beta_s, \ r > 1 \text{ and } r'>1 
\end{cases}
\end{align}

Using the definition of the operator norm (\cite{Rudin}, Definition 9.6 (c)) we bound the norm of the Hessian of the dual function,
\begin{align}
    ||\grad^{2} D(p)|| \leq \underset{s \in S}{\max} \ \beta_s \cdot \ |S|.
\end{align}
The result of the lemma then follows by application of Theorem \ref{thm:Rudin9.19}. Proof of Theorem \ref{thm:ConvergenceThm} is assured by the following Theorem, which follows from the Descent Lemma of Convex Optimization Theory \cite{BertsekasConvOpt},

\begin{theorem}[\cite{BertsekasConvOpt}]
Let $f: \mathbb{R}^{n} \rightarrow \mathbb{R}$ be a continuously differentiable function and let $X$ be a closed convex set. Assume $\grad f$ satisfies the Lipschitz condition with Lipschitz constant $L$ and consider the gradient projection iteration,
$$ x_{k+1} = P_x \big{(} x_k - \gamma \grad f (x_k)\big{)},$$
with a constant step-size $\gamma$ in the range $\big{(} 0 , \ \frac{2}{L}\big{)}$. Then every limit point $\overline{x}$ of the generated sequence $\{ x_k\}$ satisfies the optimality condition: 
$$\grad f (\overline{x})^{\text{T}} (x - \overline{x}) \geq 0, \ \forall \ x \in X.  $$
\end{theorem}

\end{proof}